\documentclass[11pt]{article}
\usepackage{theorem}
\usepackage{latexsym}
\usepackage{graphicx}
\usepackage{amssymb}
\usepackage{amsmath}
\reversemarginpar
\theoremheaderfont{\bf}
\theorembodyfont{\sl}
\newtheorem{theo}{Theorem}[section]
{\theorembodyfont{\rm} \newtheorem{defi}[theo]{Definition}}
{\theorembodyfont{\rm} }
{\theorembodyfont{\rm} \newtheorem{exa}[theo]{Example}}
{\theorembodyfont{\rm} \newtheorem{rem}[theo]{Remark}}
\newtheorem{prop}[theo]{Proposition}
\newtheorem{cor}[theo]{Corollary}

{\theorembodyfont{\rm} }
\newtheorem{lemma}[theo]{Lemma}
{\theorembodyfont{\rm}}
{\theorembodyfont{\rm}}
\newenvironment{proof}{{\sc Proof:}}{\mbox{}\hfill$\Box$\par}

\newcommand{\eqnref}[1]{~\mbox{$(${\rm \ref{#1}}$)$}}
\newcommand{\Section}[1]{\section{#1}\setcounter{equation}{0}}

\newcommand{\junk}[1]{}

\newcommand{\TS}{\textstyle}
\newcommand{\N}{{\mathbb N}}
\newcommand{\F}{{\mathbb F}}

\newcommand{\cC}{{\mathcal C}}

\newcommand{\cL}{{\mathcal L}}
\newcommand{\cS}{{\mathcal S}}
\newcommand{\cW}{{\mathcal W}}
\newcommand{\cI}{{\mathcal I}}
\newcommand{\cA}{{\mathcal A}}
\newcommand{\cB}{{\mathcal B}}
\newcommand{\cT}{{\mathcal T}}

\newcommand{\kk}{\mbox{$\underline{k}$}}

\newcommand{\rk}{\text{rk}\,}

\newcommand{\KV}{\text{KV}}

\newcommand{\edge}[1]{\mbox{$\longrightarrow$}\hspace{-1.3em}\raisebox{1.3ex}{${\scriptscriptstyle{#1}}$}\hspace{.8em}}
\newcommand{\im}{\mbox{\rm im}\,}

\newcommand{\col}{\mbox{\rm col}}

\newcommand{\T}{\mbox{$\!^{\sf T}$}}

\renewcommand{\mod}{\,{\rm mod}\,}
\newenvironment{liste}{\begin{list}{--\hfill}{\topsep0ex \labelwidth.4cm
   \leftmargin.5cm \labelsep.1cm \rightmargin0cm \parsep0ex \itemsep.6ex
   \partopsep1.4ex}}{\end{list}}
\newcounter{alp}
\newcounter{ara}
\newcounter{rom}
\newenvironment{romanlist}{\begin{list}{(\roman{rom})\hfill}{\usecounter{rom}
     \topsep0ex \labelwidth.7cm \leftmargin.7cm \labelsep0cm
     \rightmargin0cm \parsep0ex \itemsep.6ex
     \partopsep1.6ex}}{\end{list}}
\newenvironment{alphalist}{\begin{list}{(\alph{alp})\hfill}{\usecounter{alp}
     \topsep0ex \labelwidth.7cm \leftmargin.7cm \labelsep0cm
     \rightmargin0cm \parsep0ex \itemsep.6ex
     \partopsep1.6ex}}{\end{list}}

\title{Linear Tail-Biting Trellises: Characteristic Generators and the BCJR-Construction}
\date{March~23, 2010}
\author{Heide Gluesing-Luerssen\thanks{The author was partially supported by National Science Foundation
        grant \#DMS-0908379}\ \; and Elizabeth Weaver
       \\
       University of Kentucky\\
       Department of Mathematics\\
       715 Patterson Office Tower\\
       Lexington, KY 40506-0027, USA;
       \\
       $\{$heidegl,eweaver$\}$@ms.uky.edu
       }

\begin{document}
\maketitle
\noindent{\bf Abstract:}
We investigate the constructions of tail-biting trellises for linear block codes introduced by Koetter/Vardy~\cite{KoVa03}
and Nori/Shankar~\cite{NoSh06}.
For a given code we will define the sets of characteristic generators more generally than in~\cite{KoVa03}
and we will investigate how the choice of characteristic generators affects the set of resulting product trellises,
called KV-trellises.
Furthermore, we will show that each KV-trellis is a BCJR-trellis, defined in a slightly stronger sense than in~\cite{NoSh06}, and
that the latter are always non-mergeable.
Finally, we will address a duality conjecture of Koetter/Vardy by making use of a dualization technique of BCJR-trellises and prove
the conjecture for minimal trellises.



%%%%%%%%%%%%%%%%%%%%%%%%%%%%%%%%%%%%%%%%%
\Section{Introduction}\label{S-Intro}
%%%%%%%%%%%%%%%%%%%%%%%%%%%%%%%%%%%%%%%

During the last decade conventional and tail-biting trellis representations of linear block codes have gained a great deal of attention
because both types of trellises do not only reveal code structure but also may give rise to efficient decoding algorithms of Viterbi type.

As a consequence, the quest is on for constructing trellises with low complexity with respect to various measures.
For conventional trellises it is by now well-known that for a given linear block code (and with a fixed coordinate ordering)
there exists a unique minimal trellis and minimality coincides with biproperness as well as with non-mergeability; see \cite{Mu88,McE96,Fo88}.
This unique minimal trellis simultaneously minimizes the vertex count at each time as well as any other conceivable complexity measure.
Several methods exist to construct this minimal trellis, the BCJR-construction~\cite{BCJR74, McE96}, the Forney-construction based in past-future induced
realizations~\cite{Fo88} and others.
For an excellent overview on conventional trellises we refer to Vardy's survey~\cite{Va98}.

For tail-biting trellises the situation is considerably more difficult.
Due to an additional degree of freedom (the cardinality of the vertex space at time~$0$) minimal tail-biting trellises are not unique
and this applies to minimality with respect to various measures, e.~g.,
the state complexity profile, the maximum state dimension, the sum of the state dimensions as well as to all
these notions where we replace the state spaces by the edge spaces, see, for instance,~\cite{CFV99,KoVa03}.
Moreover, minimality, non-mergeability, and biproperness are not equivalent for tail-biting trellises.
The fact, however, that the complexity of a tail-biting trellis may be much lower than that of the minimal conventional trellis raised a great deal
of interest in these trellises and, in particular, in the question how to efficiently construct and to analyze tail-biting trellises;
see also~\cite{ShBe00,LiSh00} in addition to the aforementioned references.

A major breakthrough has been made by Koetter/Vardy in the excellent paper~\cite{KoVa03}.
Therein, Koetter/Vardy presented a construction from which all minimal tail-biting trellises can be derived.
Their construction accounts for all cyclic shifts of the code and results in a list of characteristic generators (codewords associated with spans),
which then comprise what they call the characteristic matrix.
It can be shown that each minimal trellis (with respect to any of the complexity measures mentioned above) is a KV-trellis, that is,
arises as the product of the elementary trellises associated with characteristic generators.
However, some care needs to be exercised at this point because not all minimal trellises of the given code can be obtained from the same list of characteristic generators.
We will present an example later on in~\ref{E-CMmintrellis}.
Furthermore, as observed in~\cite{KoVa03}, not all KV-trellises are minimal trellises.
But, as we will see later on, they are always non-mergeable.

A different construction of tail-biting trellises has been introduced by Nori/Shankar in~\cite{NoSh06}.
Their construction generalizes the well-known BCJR-construction of conventional trellises.
The additional degree of freedom, the choice of the vertex space at time~$0$, is made explicit by a displacement matrix
which can be regarded as capturing the circular past of the various trajectories (cycles) through the trellis.
In the conventional case this past is simply zero.
While the displacement matrix can be chosen arbitrarily, the best trellises arise when the matrix is based on the choice of
spans for the given basis of the code.
Later on we will restrict ourselves mainly to those displacement matrices.
Nori/Shankar~\cite{NoSh06} showed that a non-mergeable product trellis is isomorphic to the BCJR-trellis based
on the same span list.
A particularly nice feature of the tail-biting BCJR-trellises is that they naturally give rise to a dualization technique,
resulting in a trellis for the dual code with the same state complexity profile.

In this paper we will discuss the relation between the Koetter/Vardy and Nori/Shankar approach and present further properties.
We will show that BCJR-trellises (in a slightly stronger sense than in~\cite{NoSh06}) are non-mergeable and that all
KV-trellises are BCJR-trellises, hence non-mergeable.
We will further show that any one-to-one product trellis can be merged to a BCJR-trellis and merging is accomplished by taking suitable
quotients of the vertex spaces.
Finally, we will discuss a duality conjecture stated in~\cite{KoVa03}.
It deals with a particular way of dualizing a KV-trellis.
This leads to a trellis with the same state complexity profile, but for which
it is not a priori clear whether it represents the dual code.
We will show that, in general, this depends on the choice of the characteristic matrix and we will prove
that for minimal trellises there always exists a suitable choice of the dual characteristic matrix such that the conjecture holds true.
In that case the dual trellis obtained via Koetter/Vardy's characteristic matrix turns out to be the same as the BCJR-dual.
The main tool for the proof is a dualization technique going back to Mittelholzer~\cite{Mi95} and Forney~\cite{Fo01} which
amounts to dualizing the edge spaces with respect to a particular bilinear form.
As the proof will show, for minimal trellises the resulting edge space dual coincides with the BCJR-dual.
We will conclude the paper with a list of open problems.

%%%%%%%%%%%%%%%%%%%%%%%%%%%%%%%%%%%%%%%%%
\Section{Preliminaries}
%%%%%%%%%%%%%%%%%%%%%%%%%%%%%%%%%%%%%%%%%
In this section we introduce the basic notions for tail-biting trellises.
Throughout the paper let~$\F$ be a finite field.

A {\sl tail-biting trellis\/} $T=(V,E)$ over~$\F$ of depth~$n$ is a directed edge-labeled graph
with the property that the vertex set~$V$ partitions into $n$ disjoint sets
$V=V_0\cup V_1\cup\ldots\cup V_{n-1}$
such that every edge in~$T$ that starts in~$V_i$ ends in $V_{i+1\mod n}$ and where the edges are labeled with
field elements from~$\F$.
Thus, the edge set~$E$ decomposes into
$E=\bigcup_{i=0}^{n-1}E_i,$ where $E_i$ is the set of edges starting in~$V_i$ and ending in~$V_{i+1\mod n}$.
A {\sl cycle\/} in~$T$ is a closed path of length~$n$ in the trellis.
We will always assume that the cycles start and end (at the same point) in~$V_0$.
If~$|V_0|=1$, the trellis is called {\sl conventional\/} (we will not consider conventional trellises where $|V_0|>1$).
We call the trellis {\sl reduced\/} if every vertex and every edge appears in at least one cycle.
The trellis is called {\sl biproper\/} if any two edges starting at the same vertex or ending at the same vertex are labeled distinctly.
The interval $\cI:=\{0,\ldots,n-1\}$ is the {\sl time axis\/} of the trellis.

In the sequel we will always identify the edges with their triples consisting of starting point, label, and endpoint point.
Thus, the edge sets are given by
\[
   E_i=\{(v,a,\hat{v})\mid \text{ there exist an edge } v\edge{\;a}\,\hat{v} \text{ where }v\in V_i,\,\hat{v}\in V_{i+1},\,a\in\F\},\,i\in\cI.
\]
Hence~$E_i\subseteq V_i\times\F\times V_{i+1}$ and $(V,E)$ completely specifies the trellis.
The edge sets have also been called trellis section~\cite{CFV99,FT93} or local constraints~\cite{Fo01}.

Notice that we have to compute modulo~$n$ on the time axis~$\cI$.
If the trellis is conventional, however, there is no need in doing so because in that case all paths are cycles
and the endpoints of edges in~$E_{n-1}$ do not carry any information.

Each tail-biting trellis gives rise to its edge-label code consisting of the edge-label sequences of
all cycles in~$T$.
We will need to make this and other related notions precise.

%%%%%%%%%%%%%%%%%%%%
\begin{defi}\label{D-TBTbasics}
Let~$T=(V,E)$ be a tail-biting trellis over~$\F$ of depth~$n$ and put
$\cW:=V_0\times\F\times\ldots\times V_{n-1}\times\F$.
\begin{alphalist}
\item The {\sl label code\/} of~$T$ is defined as
      \begin{align*}
        &\cS(T):=\{(v_0,c_0,v_1,\ldots,v_{n-1},c_{n-1})\in\cW
        \mid\\[.5ex]
        &\hspace*{10em}\text{ there exist a cycle }v_0\edge{c_0}v_1\edge{c_1}\ldots v_{n-1}\edge{\!\!\!\!c_{n-1}}\!\!v_0\text{ in~$T$}\}.
      \end{align*}
      The trellis~$T$ is called {\sl linear\/} if~$T$ is reduced and the vertex sets $V_i,i\in\cI,$ are vector spaces over~$\F$ and
      $\cS(T)$ is a subspace of~$\cW$.
      In this case the edge spaces $E_i,i\in\cI,$ are subspaces of $V_i\times\F\times V_{i+1}$ because they are projections of
      $\cS(T)$.
\item The {\sl edge-label code\/} of~$T$ is defined as
      \[
         \cC(T):=\{(c_0,\ldots,c_{n-1})\in\F^n\mid
         \text{ there exist a cycle }v_0\edge{c_0}v_1\edge{c_1}\ldots v_{n-1}\edge{\!\!\!\!c_{n-1}}\!\!v_0\text{ in~$T$}\}.
      \]
      We say that~$T$ {\sl represents the code\/}~$\cC\subseteq\F^n$ if $\cC(T)=\cC$.
      Note that if~$T$ is linear then $\cC(T)$ is a linear block code.
      The trellis~$T$ is called {\sl one-to-one\/} if there exists a bijection between $\cC(T)$ and the cycles in~$T$.
      That is, every word in $\cC(T)$ appears as the edge label sequence of exactly one cycle.
\item If~$T$ is linear the {\sl state complexity profile\/} (SCP) is defined as $(s_0,\ldots,s_{n-1})$, where $s_i=\dim V_i$.
      The {\sl edge complexity profile\/} (ECP) is defined as $(e_0,\ldots,e_{n-1})$, where $e_i=\dim E_i$.
\item Two trellises $T=(V,E)$ and $T'=(V',E')$ are called {\sl isomorphic\/} if there exists a bijection $\phi:V\longrightarrow V'$ such that
      $\phi(V_i)=V'_i$ for all $i\in\cI$ and such that $(v,\,a,\,w)\in E_i$ if and only if $(\phi(v)\,\,a,\,\phi(w))\in E'_i$.
      If~$T$ and~$T'$ are linear, then we also require $\phi|_{\TS V_i}:V_i\longrightarrow V'_i$ be isomorphisms.
      Isomorphic trellises represent the same code.
\item Two trellises $T=(V,E)$ and $T'=(V',E')$ are called {\sl structurally isomorphic\/} if there exists a bijection
      $\phi:V\longrightarrow V'$ such that $\phi(V_i)=V'_i$ for all $i\in\cI$ and such that for all $(v,w)\in V_i\times V_{i+1}$
      the number of edges from~$v\in V_i$ to $w\in V_{i+1}$ equals the number of edges from $\phi(v)\in V'_i$ to $\phi(w)\in V'_{i+1}$
      (thus we disregard the edge labels).
      Again, if~$T$ and~$T'$ are linear we also require $\phi|_{\TS V_i}:V_i\longrightarrow V'_i$  be isomorphisms.
\end{alphalist}
\end{defi}
%%%%%%%%%%%%%%%%%%%%%%%%%%%%%%%

Throughout the paper we will only deal with linear block codes and linear (hence reduced) trellis representations.
Observe that, in general, structurally isomorphic trellises do not represent the same code.
However, we will only be interested in this notion for trellises of a fixed code.

We wish to briefly touch upon all these notions from the viewpoint of behavioral systems theory.
A trellis~$T$ can be regarded as a dynamical system with latent variables in the sense of \cite[Def.~1.3.4]{PoWi98},
where the latent variables represent the vertex sequence and satisfies the axiom of state at each time~$i$, see~\cite[Def.~1.3]{Wi89}.
In particular, the vertex space~$V_i$ is the state space at time~$i$.
In this context the edge spaces, consisting of all state-output-next state triples, are called the evolution law~\cite[Def.~1.4]{Wi89}.
Completeness as defined in~\cite[p.~184]{Wi89}, a crucial notion in behavioral systems theory relating
local and global behavior does not play a role in our setting since, due to the finite time axis, each system is complete.
Notice that the system is, in general, highly time-varying since even the state space dimensions depend on time.
The label code $\cS(T)$ is --- up to the ordering of labels and vertices --- the full behavior\footnote{Strictly speaking,
the full behavior is the space of all paths in~$T$ and not just the cycles.
}
while the edge-label code~$\cC(T)$ is the manifest behavior, see again \cite[Def.~1.3.4]{PoWi98}.
If we consider the code~$\cC$ as the given system, which, as usual in behavioral systems theory, is identified with its set of
trajectories (the codewords), then~$T$ is simply a state-space realization of~$\cC$ if $\cC(T)=\cC$.
Furthermore, isomorphic trellises are equivalent systems in the sense of \cite[p.~205]{Wi89}.
They merely differ by the labeling of the latent variable (the vertices).
The system theoretic notion of minimal realizations now becomes minimal trellises, that is, minimizing
the amount of data needed for a realization.
While for conventional trellises this is well-understood and amounts to the unique minimal trellis (up to isomorphism) with a
variety of universal properties, the situation is much more complicated for tail-biting trellises.
This leads to various notions of minimality for trellises of which only the following one will be relevant for us.
For other notions see for instance \cite{Mu88,McE96} or \cite[p.~2085]{KoVa03}.
In this context we will also introduce the notion of mergeability of trellises (In system theoretic language merging as defined below
is also referred to as lumping states, see \cite[p.~207]{Wi89}).

%%%%%%%%%%%%%%%%%%%%%%%%
\begin{defi}\label{D-TBTmin}
Let $\cC\subseteq\F^n$ be a linear block code and~$T=(V,E)$ be a linear trellis of~$\cC$.
\begin{alphalist}
\item $T$ is called {\sl minimal\/} if there exists no linear trellis $T'=(V',E')$ of~$\cC$ such that
      $|V'_i|\leq |V_i|$ for all $i\in\cI$ and $|V'_j|<|V_j|$ for some $j\in\cI$.
\item $T$ is called {\sl non-mergeable\/} if for any $i\in\cI$ merging any two vertices $v_1,\,v_2\in V_i$
      results in a trellis $\hat{T}$ that does not represent $\cC$.
      Here merging
      means replacing the two vertices $v_1,\,v_2$ in the trellis by a single vertex $\hat{v}\in V_i$
      and replacing  each edge of the form $(v_j,\,a,\,w)\in E_i$, resp. $(w,\,a,\, v_j)\in E_{i-1}$, where $j=1,2$,
      by the edge $(\hat{v},\,a,\,w)\in E_i$, resp. $(w,\,a,\,\hat{v})\in E_{i-1}$.
\end{alphalist}
\end{defi}
%%%%%%%%%%%%%%%%%%%%%%%%

Due to our restriction to linear trellises it seems necessary to also address the issue whether
merging leads to a potentially non-linear trellis.
It turns out however, that mergeability is compatible with the linear structure and hence
merging can always be accomplished such that the resulting trellis is linear again.

It is not a priori clear whether minimal (linear) trellises exist for a given linear code.
For conventional trellises this is indeed the case as is well-known since a long time and leads to minimizing many other
measures as well; see \cite[Thm.~4.26]{Va98}.
Moreover, all minimal conventional trellises are isomorphic and even minimal in the larger class of {\sl all\/}
trellises, not necessarily linear.
Furthermore, for conventional trellises minimality coincides with non-mergeability.
A nice way of constructing the unique minimal conventional trellis of a given code, using realization theory for
dynamical systems, has been given by Forney in~\cite{Fo88}.
For tail-biting trellises minimality has been studied in much detail by Calderbank/Forney/Vardy in~\cite{CFV99} and
Koetter/Vardy in~\cite{KoVa03}.
In particular, it has been shown that linear minimal tail-biting trellises exist (the minimal conventional trellis is one of them)
and there exist non-isomorphic ones.
This is actually easy to see by constructing the minimal conventional trellis for each cyclic shift of the given code and
then shifting back that trellis; see \cite[p.~2084]{KoVa03}.
This results in many (non-isomorphic) minimal trellises and in each one at least one vertex space~$V_i$ is trivial.
However, in general there exist other minimal trellises (that is, for which none of the vertex spaces is trivial);
see for instance \cite[Exa.~6.1]{KoVa03} for the $(8,4,4)$ Hamming code, which also appeared in \cite[p.~1449]{CFV99}.
Moreover, different from the situation for conventional trellises, minimality is stronger than non-mergeability for tail-biting trellises;
see for instance \cite[Exa.~3.3]{KoVa03}.
The main result of~\cite{KoVa03} is a construction from which all minimal tail-biting trellises of
a given code can be derived.
Our paper is mainly devoted to investigating this construction, which we will therefore introduce in the next section.

Finally, let us fix the following notation.
For a matrix $M\in\F^{m\times n}$ we denote
its row span as $\im M$, thus $\im M:=\{\alpha M\mid \alpha\in\F^m\}$.
If not otherwise stated,
$\cC\subseteq\F^n$ denotes a $k$-dimensional code of length~$n$ over the finite field~$\F$
with the following generator and parity check matrices:
\begin{equation}\label{e-Gdata}
  \cC=\im G,%:=\{\alpha G\mid \alpha\in\F^k\},
    \text{ where }
     G=(g_{lj})=\begin{pmatrix}g_1\\\vdots\\g_k\end{pmatrix}=\begin{pmatrix}G_0&\ldots&G_{n-1}\end{pmatrix}\in\F^{k\times n},
\end{equation}
and
\begin{equation}\label{e-Hdata}
   \cC=\ker H\T:=\{c\in\F^n\mid cH\T=0\},
   \text{where }H\T=\begin{pmatrix}H_0\\\vdots\\ H_{n-1}\end{pmatrix}\in\F^{n\times(n-k)}.
\end{equation}
Hence $G_i\in\F^k$ and $H_i\in\F^{n-k}$ are the columns of~$G$ and rows of~$H\T$, respectively, while
$g_1,\ldots,g_k\in\F^n$ are the rows of~$G$.
For simplicity we will also assume that the support of~$\cC$ is~$\cI=\{0,\ldots,n-1\}$, that is, no column of~$G$ is zero.
Finally, it will be convenient to have a notion for the indicator function of a subset $\cA\subseteq\cI$.
Thus, we define $I^{\cA}\in\F^n$, where $I^{\cA}_j=1$ if $j\in\cA$ and $I^{\cA}_j=0$ else.

%%%%%%%%%%%%%%%%%%%%%%%%%%%%%%%%%%%%%%%%%%
\Section{Product Trellises and KV-Trellises}\label{S-KV}
%%%%%%%%%%%%%%%%%%%%%%%%%%%%%%%%%%%%%%%%%%
We begin with recalling the construction of tail-biting trellises based on products of elementary trellises.
The underlying notions and constructions are well-known, see \cite{KschSo95} for the conventional case and \cite{CFV99,KoVa03}
for the tail-biting case.
Thereafter, we will introduce the notion of a characteristic pair for a given code, from which in essence all minimal tail-biting
trellises can be derived.
This crucial and very fruitful concept goes back to Koetter/Vardy~\cite{KoVa03}, who in turn based it on previous work by
Kschischang/Sorokine~\cite{KschSo95}.
The relation between minimal trellises and characteristic matrices will be discussed in detail.

Due to the cyclic structure of the time axis~$\cI$ for tail-biting trellises, the following interval
notation has proven to be very convenient.
For $a,\,b\in\cI$ we define
\begin{equation}\label{e-interval}
  \left.\begin{array}{l}
   [a,\,b]=\left\{\begin{array}{ll}\{a, a+1,\ldots,b\}&  \text{if }a\leq b,\\[.3ex]
               \{a,a+1,\ldots,n-1,0,1,\ldots,b\}& \text{if }a > b,
           \end{array}\right. \\[2.4ex]
   (a,\,b]=[a,b]\backslash\{a\}\hspace*{11.3em} \text{for all }a,\,b.
  \end{array}\qquad\quad\right\}
\end{equation}
We call the intervals $(a,\,b]$ and $[a,\,b]$ {\sl linear\/} if $a\leq b$ and {\sl circular\/} else.
Notice that $(a,a]=\emptyset$.
It is easy to see that
\begin{equation}\label{e-compint}
   \cI\,\backslash\,(a,\,b]=(b,\,a] \text{ for all }a\not=b,
\end{equation}
hence the complement of a nonempty linear interval is circular and vice versa.
Observe also that $0\not\in(a,\,b]\Longleftrightarrow (a,\,b]$ is linear.

%%%%%%%%%%%%%%%%%%%%%%%%%%%
\begin{defi}\label{D-vectorspan}
For a nonzero vector $c=(c_0,\ldots,c_{n-1})\in\F^n$ we call the interval $(a,\,b]$ a {\sl span\/} of~$c$ if
$c_a\not=0\not=c_b$ and $c_j=0$ for $j\not\in[a,\,b]$.
In this case we will also call $[a,\,b]$ the {\sl closed span\/} of~$c$.
\end{defi}
%%%%%%%%%%%%%%%%%%%%%%%%%%
For instance, the vector $v=(0,1,3,0,2,2)\in\F_5^6$ has the spans $(1,5],\,(2,1],\,(4,2]$, and $(5,4]$.
Notice that every nonzero vector has a unique linear span.
It turns out that not including the starting point~$a$ in the span of~$x$ is very convenient for later notation;
see, for instance, the next definition.

%%%%%%%%%%%%%%%%%%%%%%%%%%%%%%
\begin{defi}\label{D-EltTrellis}
Let $c=(c_0,\ldots,c_{n-1})\in\F^n$ and $(a,b]$ be a span of~$c$.
The {\sl elementary trellis for the pair\/} $\big(c,(a,b]\big)$ is defined as $T_{c,(a,b]}:=(V,E)$, where
the vertex sets and edge sets $V=\cup_{i=0}^{n-1}V_i$ and $E=\cup_{i=0}^{n-1}E_i$
are as follows:
\\
if $a=b$ then,
\[
  V_i:=\{0\}\text{ for all }i\in \cI\text{ and }
  E_i:=\left\{\begin{array}{ll}\{(0,0,0)\},&\text{if }i\not=a\\[.3ex]
        \{(0,\alpha c_i,0)\mid \alpha\in\F\},&\text{if }i=a;\end{array}\right.
\]
if $a\not=b$, then
\[
    V_i:=\left\{\begin{array}{ll}0,&\text{if } i\not\in (a,\,b]\\[.3ex] \F,&\text{if } i\in (a,\,b]
         \end{array}\right.
    \text{ and }
    E_i:=\left\{\begin{array}{ll}\{(0,0,0)\},&\text{if }i\not\in[a,\,b]\\[.3ex]
                 \{(0,\alpha c_i,\alpha)\mid \alpha\in\F\},&\text{if }i=a\\[.3ex]
                 \{(\alpha,\alpha c_i,0)\mid \alpha\in\F\},&\text{if }i=b\\[.3ex]
                 \{(\alpha,\alpha c_i,\alpha)\mid \alpha\in\F\},&\text{if }i\in(a,b-1]
         \end{array}\right.
\]
Using the indicator function $I^{(a,b]}$, the above vertex and edge spaces can be written as $V_i=\im (I^{(a,b]}_i)\subseteq\F$ and
$E_i=\im(I^{(a,b]}_i,c_i,I^{(a,b]}_{i+1})\subseteq V_i\times\F\times V_{i+1}$.
\end{defi}
%%%%%%%%%%%%%%%%%%%%%%%%%%%%%%
Note that $T_{c,(a,b]}$ is a one-to-one linear trellis and $\cC(T)=\im c=\{\alpha c\mid\alpha\in\F\}$, the one-dimensional code generated by~$c$.

With the aid of the product construction as introduced in \cite{KschSo95,KoVa03} we can derive trellises for higher-dimensional codes.
Recall that the product $\hat{T}:=T\times T'$ of two trellises $T=(V,E)$ and $T'=(V',E')$  of depth~$n$ over~$\F$ is defined as
the trellis $\hat{T}=(\hat{V},\,\hat{E})$, where
$\hat{V}_i:=V_i\times V'_i$ and
$\hat{E}_i=\big\{\big((v,v'),a+a',(w,w')\big)\,\big|\, (v,a,w)\in E_i,\,(v',a',w')\in E'_i\big\}$
for all $i\in\cI$.
The product trellis~$\hat{T}$ satisfies $\cC(\hat{T})=\cC(T)+\cC(T')$.

The most important application of this construction is the product of elementary trellises leading to
tail-biting trellis representations of linear block codes.
Recall the data for the given code~$\cC\subseteq\F^n$ as given in\eqnref{e-Gdata}.

%%%%%%%%%%%%%%%%%%%%%%%%%%%
\begin{theo}[\cite{KschSo95,KoVa02,KoVa03}]\label{T-KVtrellis}
Let
\begin{equation}\label{e-S}
   \cS:=\big[(a_l,b_l],\,l=1,\ldots,k\big]
\end{equation}
be a {\sl span list for~$G$}, that is, $(a_l,b_l]$ is a span (linear or circular) for the
row~$g_l,\,l=1,\ldots,k$.
Then product trellis $T_{G,\cS}:=T_{g_1,(a_1,b_1]}\times\ldots\times T_{g_k,(a_k,b_k]}$
represents the code~$\cC$, that is, $\cC(T_{G,\cS})=\cC$.
Moreover,~$T_{G,\cS}$ is linear (in particular reduced) and one-to-one.
The vertex spaces of~$T_{G,\cS}$ are given by
\begin{equation}\label{e-MMat}
  V_i=\im M_i,\text{ where }M_i=\begin{pmatrix}\!\!\mu^1_{i}& & \\ &\ddots& \\& &\mu^k_{i}\!\!\end{pmatrix}
  \in\F^{k\times k} \text{ and }
  \mu^l_{i}=\left\{\begin{array}{ll}\!\!1,&\text{if }i\in(a_l,b_l]\\ \!\!0,&\text{if }i\not\in(a_l,b_l]\end{array},\right.
\end{equation}
while the edge spaces are given by $E_i=\im(M_i, G_i,M_{i+1})$ for $i\in\cI$.
\end{theo}
%%%%%%%%%%%%%%%%%%%%%%%%%%%%
\label{PageFootnote}
Observe that $\mu^l=I^{(a_l,b_l]}$ is the indicator function of $(a_l,b_l]$.
It is straightforward to see that, more generally, a product trellis
$T:=T_{g_1,(a_1,b_1]}\times\ldots\times T_{g_k,(a_k,b_k]}$ is one-to-one if and only if $g_1,\ldots,g_k$ are linearly
independent.\footnote{There is one extreme case to exclude from this statement: if some spans are of the form
$(a_l,a_l]$ then~$T$ is one-to-one if and only if the set~${\cal G}$ is linearly independent, where ${\cal G}$ is
obtained by taking for each span of the form $(a_l,a_l]$ only one generator
and taking all generators with non-empty span.
Since in our considerations $g_1,\ldots,g_k$ are always linearly independent this case will not be relevant for us.}

In the sequel we will only consider one-to-one product trellises (with the exception of Remark~\ref{R-nonmerg}).
In addition to this restriction our product trellises are confined even further when compared to the class considered by
Koetter/Vardy~\cite{KoVa03}.
Due to Definition~\ref{D-vectorspan} we only consider ``shortest spans'' for the generators, that is, the vector
is nonzero at the endpoints of the span (in other words, the closed span is a smallest linear or circular interval containing the support).
This restriction is justified by the fact that an elementary trellis based on a non-shortest span is mergeable and thus
so is the product trellis, see~\cite[Lemma~4.3]{KoVa03}.
Of course, even with shortest spans a product trellis may be mergeable; we will see examples later on.

The following shift property of product trellises is a direct consequence.
Recall that we compute with indices modulo~$n$. This, of course, also applies to the span lists.

%%%%%%%%%%%%%%%%%%%%%%%%%%%%
\begin{rem}\label{R-shiftKV}
Let $\sigma:\F^n\longrightarrow\F^n,\ (c_0,\ldots,c_{n-1})\longmapsto(c_1,\ldots,c_{n-1},c_0)$ be
the left shift on~$\F^n$ and let $G^*\in\F^{k\times n}$ be the matrix consisting of the shifted
rows $\sigma(g_l),\,l=1,\ldots,k$.
If~$\cS$ as in\eqnref{e-S} is a span list for~$G$ then
$\cS^*=\big[(a_l-1,\, b_l-1],l=1,\ldots,k\big]$ forms a span list for~$G^*$.
Furthermore, with the notation as in\eqnref{e-MMat} the vertex spaces of the product trellis $T_{G^*,\cS^*}$ are given by
$V_i^*=\im M_{i+1}$ and the edge spaces are given by $E_i^*=\im(M_{i+1},G_{i+1},M_{i+2})$
for $i\in\cI$.
\end{rem}
%%%%%%%%%%%%%%%%%%%%%%%%%%%%

For the product trellises $T_{G,\cS}$ it is not hard to give formulas for the SCP and the ECP in terms of the span list~$\cC$.
Let $\underline{k}:=\{1,\ldots,k\}$.

%%%%%%%%%%%%%%%%%%%%%%%%%%%
\begin{prop}\label{P-SCPECP}
Let~$\cS$ and $T_{G,\cS}$ be as in Theorem~\ref{T-KVtrellis}.
Put $\cA:=\{a_1,\ldots,a_k\}$ and $\cB=\{b_1,\ldots,b_k\}$. Moreover, for each $i\in\cI$ define
\begin{equation}\label{e-Li}
   \cL_i:=\{l\in\underline{k}\mid i\in(a_l,b_l]\}.
\end{equation}
Then the SCP of~$T_{G,\cS}$  is given by $(s_0,\ldots,s_{n-1})$, where $s_i=|\cL_i|$.
The ECP $(e_0,\ldots,e_{n-1})$ satisfies the following.
\begin{alphalist}
\item If $a_1,\ldots,a_k$ are distinct, then $e_i=s_i+I^{\cA}_i$ for $i\in\cI$.
\item If $b_1,\ldots,b_k$ are distinct, then $e_i=s_{i+1}+I^{\cB}_i$ for $i\in\cI$.
\item Let $a_1,\ldots,a_k$ be distinct and $b_1,\ldots,b_k$ be distinct. Then the SCP satisfies
      \[
          s_{i+1}=\left\{\begin{array}{ll}s_i,&\text{ if }i\in\cA\cap\cB\text{ or }i\not\in\cA\cup\cB,\\
                                                s_i+1,&\text{ if }i\in\cA\,\backslash\,\cB,\\
                                                s_i-1,&\text{ if }i\in\cB\,\backslash\,\cA,\end{array}\right.
      \]
      and the trellis $T_{G,\cS}$ is biproper.
\end{alphalist}
\end{prop}
%%%%%%%%%%%%%%%%%%%%%%%%%%%%

Observe that if all spans $(a_l,b_l]$ are linear, then the condition $a_1,\ldots,a_k$ and $b_1,\ldots,b_k$ both
being distinct is equivalent to~$G$ being an MSGM in the sense of \cite[Def.~6.2, Thm.~6.11]{McE96}
(also called ``sets of atomic generators'' in \cite{KschSo95} or ``shortest basis'' in~\cite{Fo09}).
Moreover, in this case $T_{G,\cS}$ is the unique minimal conventional trellis of~$\cC$ and
the formulas given in~(a) and~(b) coincide with those known for conventional trellises; see, for instance,
\cite[p.~1080]{McE96}.
However, it is well-known that if not all spans are linear the (tail-biting) product trellis is in general not minimal; see,
for instance, Example~\ref{E-BCJRnotone} later on.

\begin{proof}
The first statement about the SCP is immediate from the definition of the vertex spaces $V_i=\im M_i$ in\eqnref{e-MMat}.
For the ECP, recall that by definition of the edge spaces we have $e_i=\rk(M_i,G_i,M_{i+1})$.
\\
(a) We proceed as follows.
For a matrix~$M\in\F^{k\times r}$ denote by $\col(M)\subseteq\F^k$ its column space and by $\col(M,t)$ its $t$-th column.
Fix~$i\in\cI$.
Let $\hat{e}_i:=\rk(M_i,G_i)$.
In a first step we will show that $\hat{e}_i=s_i+I^{\cA}_i$ and in a second step we will prove that $e_i=\hat{e}_i$.
For the first step, notice that $\hat{e}_i=s_i$ if $G_i\in\col(M_i)$ and $s_i+1$ else.
Moreover, the definition of the matrices~$M_i$ implies
\[
  G_i\in\col(M_i)\Longleftrightarrow g_{li}=0\text{ for all }l\not\in\cL_i.
\]
Since $(a_l,b_l]$ is a span for the row~$g_l$ one has $g_{li}=0$ for all $i\not\in[a_l,\,b_l]$ and $g_{l,a_l}\not=0$.
But then the above along with $(a_l,\,b_l]\cup\{a_l\}=[a_l,b_l]$ shows that $G_i\in\col(M_i)\Longleftrightarrow i\not\in\cA$.
All this proves $\hat{e}_i=s_i+I^{\cA}_i$.
In order to establish $e_i=\hat{e}_i$ it remains to show that $\col(M_{i+1})\subseteq\col(M_i,G_i)$.
From\eqnref{e-MMat} we know that $\col(M_{i+1},l)=0$ if $l\not\in\cL_{i+1}$ and $\col(M_{i+1},l)=\varepsilon_l$ if $l\in\cL_{i+1}$,
where~$\varepsilon_l\in\F^k$ denotes the $k$-th standard basis vector.
Hence fix~$l\in\cL_{i+1}$. If $l\in\cL_i$, then $\col(M_{i+1},l)=\varepsilon_l=\col(M_i,l)\in\col(M_i,G_i)$ and we are done.
If $l\not\in\cL_i$, then $i\not\in(a_l,b_l]$ and thus $\col(M_i,l)=0$.
But then the assumption $i+1\in(a_l,b_l]$ implies that $i=a_l$ and thus $g_{li}\not=0$.
In order to examine the other entries of the column vector~$G_i$ let $r\in\underline{k}\,\backslash\{l\}$.
If $r\in\cL_i$ then $\col(M_i,r)=\varepsilon_r$.
If $r\not\in\cL_i$ then $i\not\in(a_r,b_r]$ and since $i=a_l$, which is distinct from~$a_r$, this implies
$i\not\in[a_r,b_r]$, hence $g_{ri}=0$.
All this shows that $G_i-\sum_{r=1}^k g_{ri}\col(M_i,r)=g_{li}\varepsilon_l=g_{li}\col(M_{i+1},l)$.
As a consequence, $\col(M_{i+1},l)\in\col(M_i,G_i)$ and therefore $e_i=\rk(M_i,G_i,M_{i+1})=\rk(M_i,G_i)=s_i+I^{\cA}_i$.
\\
(b) We proceed similarly using the rank $\tilde{e}_i:=\rk(G_i,M_{i+1})$. This time we have
\begin{align}
  G_i\in\col(M_{i+1})&\Longleftrightarrow g_{li}=0\text{ for all $l$ such that }i+1\not\in (a_l, b_l] \nonumber\\
                     &\Longleftrightarrow g_{li}=0\text{ for all $l$ such that }i\not\in (a_l-1, b_l-1]. \label{e-GiMi+1}
\end{align}
The latter semi-open interval is $[a_l,b_l-1]$ if $a_l\not=b_l$ and empty otherwise.
Since $g_{l,b_l}\not=0$\eqnref{e-GiMi+1} implies the equivalence $G_i\in\col(M_{i+1})\Longleftrightarrow i\not\in\cB$.
Hence $\tilde{e}_i=s_{i+1}+I^{\cB}_i$ and it remains to show that $\col(M_i)\subseteq\col(G_i,M_{i+1})$.
Fix $l\in\cL_i$. Then $\col(M_i,l)=\varepsilon_l$.
If $l\in\cL_{i+1}$ then $\col(M_{i+1},l)=\varepsilon_l$ and we are done.
If $l\not\in\cL_{i+1}$, then $i+1\not\in(a_l,b_l]$, but $i\in(a_l,b_l]$, thus $i=b_l$.
In this case $\col(M_{i+1},l)=0$ and $g_{li}\not=0$.
For the other entries of~$G_i$ pick $r\in\underline{k}\,\backslash\{l\}$.
If $r\in\cL_{i+1}$ then $\col(M_{i+1},r)=\varepsilon_r$.
If $r\not\in\cL_{i+1}$, then $i+1\not\in(a_r,b_r]$, hence $i\not\in(a_r-1,b_r-1]$.
But since $i=b_l$ is different from~$b_r$ this implies $i\not\in[a_r,b_r]$.
Hence $g_{ri}=0$.
All this shows that $G_i-\sum_{r=1}^k g_{ri}\col(M_{i+1},r)=g_{li}\varepsilon_l=g_{li}\col(M_i,l)$
and therefore $\col(M_i,l)\in\col(G_i,M_{i+1})$.
In particular, $\rk(M_i,G_i,M_{i+1})=\rk(G_i,M_{i+1})$.
\\
(c) The identities for~$s_i$ follow from equating the expressions for~$e_i$ in~(a) and~(b).
Furthermore, in~(a) and~(b) we saw that $\rk(M_i,G_i)=\rk(G_i,M_{i+1})=\rk(M_i,G_i,M_{i+1})$.
This in turn implies $\ker(M_i,G_i)=\ker(G_i,M_{i+1})=\ker(M_i,G_i,M_{i+1})$, from
which the biproperness follows.
\end{proof}

At this point we wish to show a path property of product trellises that will prove useful later on when discussing mergeability.
The property holds true under a certain condition on the vertex spaces, satisfied by product trellises $T_{G,\cS}$, but also by
the BCJR-trellises studied in the next section.
In system theoretic language this path property means that the trellis is point controllable (in~$n$ steps)
in the sense of~\cite[p.~188]{Wi89}.

%%%%%%%%%%%%%%%%%%%%%%%%%%%%%%%%
\begin{lemma}\label{L-pathv0}
Let~$\cC=\im G$ and let the span list~$\cS$ be as in\eqnref{e-S}.
Let $T=(V,E)$ be a trellis for~$\cC$ such that the vertex and edge spaces are given by
$V_i=\im M_i$ and $E_i=\im(M_i,G_i,M_{i+1})$ for certain matrices $M_i\in\F^{k\times r}$.
Let the sets~$\cL_i$ be defined as in\eqnref{e-Li} and suppose that for all $i\in\cI$ the $l$-th row
of~$M_i$ is zero whenever $l\not\in\cL_i$.
\\
Then for every $v\in V_0=\im M_0$ there exists a path through the trellis~$T$ that starts at~$v\in V_0$ and ends at $0\in V_0$.
Precisely, there exist vectors $\alpha^{(0)},\ldots,\alpha^{(n-1)}\in\F^k$ such that
\[
  \alpha^{(0)} M_0=v,\quad \alpha^{(i)}-\alpha^{(i+1)}\in\ker M_{i+1}\text{ for }i=0,\ldots,n-2,\quad \alpha^{(n-1)}M_0=0.
\]
\end{lemma}
%%%%%%%%%%%%%%%%%%%%%%%%%%%%%%%
\begin{proof}
In the sequel, for any vector $\alpha\in\F^k$ let~$\alpha_l$ denote the $l$-th coordinate of~$\alpha$.
By assumption
\begin{equation}\label{e-kerMj}
  \{\alpha\in\F^k\mid \alpha_l=0\text{ for }l\in\cL_i\} \subseteq\ker M_i\text{ for all }i\in\cI.
\end{equation}
Let $\alpha^{(0)}\in\F^k$ such that $\alpha^{(0)} M_0=v$.
By assumption on the rows of~$M_0$ we may assume without loss of generality
\begin{equation}\label{e-alpha0}
    \alpha^{(0)}_l=0\text{ for }l\not\in\cL_0.
\end{equation}
Define $\alpha^{(1)}\in\F^k$ via
\begin{equation}\label{e-alpha1}
   \alpha^{(1)}_l=\left\{\begin{array}{ll} \alpha^{(0)}_l,&\text{if }l\in\cL_1\cap\cL_0\\ 0,&\text{if }l\not\in\cL_1\cap\cL_0.\end{array}\right.
\end{equation}
Then\eqnref{e-alpha0} and\eqnref{e-alpha1} show that $\alpha^{(1)}_l-\alpha^{(0)}_l=0$ for $l\in\cL_1$ and thus
$\alpha^{(1)}-\alpha^{(0)}\in \ker M_1$.
Assume now we proceed in this way and have constructed $\alpha^{(j)}\in\F^k$ such that
\begin{equation}\label{e-alphaj}
   \alpha^{(j)}_l=\left\{\begin{array}{ll} \alpha^{(j-1)}_l,&\text{if }l\in\bigcap_{i=0}^{j} \cL_i\\ 0,&\text{if }l\not\in\bigcap_{i=0}^{j} \cL_i
               \end{array}\right.
\end{equation}
and such that $\alpha^{(j)}-\alpha^{(j-1)}\in\ker M_j$.
Then put
\begin{equation}\label{e-alphaj1}
   \alpha^{(j+1)}_l=\left\{\begin{array}{ll} \alpha^{(j)}_l,&\text{if }l\in\bigcap_{i=0}^{j+1} \cL_i\\ 0,&\text{if }l\not\in\bigcap_{i=0}^{j+1} \cL_i.
               \end{array}\right.
\end{equation}
Then $(\alpha^{(j+1)}-\alpha^{(j)})_l=0$ for all $l\in\bigcap_{i=0}^{j+1} \cL_i$ by\eqnref{e-alphaj1} and
$\alpha^{(j+1)}_l=\alpha^{(j)}_l=0$ for all $l\in \cL_{j+1}\backslash \bigcap_{i=0}^{j} \cL_i$ by\eqnref{e-alphaj} and\eqnref{e-alphaj1}.
Hence $(\alpha^{(j+1)}-\alpha^{(j)})_l=0$ for all $l\in\cL_{j+1}$ and thus $\alpha^{(j+1)}-\alpha^{(j)}\in\ker M_{j+1}$.
Thus, proceeding with this construction we find $\alpha^{(n-1)}\in\F^k$ such that $\alpha^{(n-1)}-\alpha^{(n-2)}\in\ker M_{n-1}$ and
\[
  \alpha^{(n-1)}_l=0\text{ for } l\not\in\bigcap_{i=0}^{n-1} \cL_i.
\]
But then $\alpha^{(n-1)}=0$ because $\bigcap_{i=0}^{n-1} \cL_i=\{l\in\kk\mid i\in(a_l,\,b_l]\text{ for all }i\in\cI\}=\emptyset$,
where the latter is due to the fact that $(a,\,b]\not=\cI$ by definition of half-open intervals in\eqnref{e-interval}.
Hence $\alpha^{(n-1)} M_n=0$ as desired.
\\
Observe that the proof actually shows that there is even a path from $v\in V_0$ to $0\in V_{n-1}$ because $\alpha^{(n-1)}=0$.
\end{proof}

A main result of~\cite{KoVa03} is the construction of a certain characteristic matrix for a given code from
which all minimal trellises can be derived.
We will define the characteristic matrix slightly differently,
namely based on certain properties rather than on the outcome of a particular procedure.
This will facilitate later considerations.
The existence of this object will then follow from the algorithm derived in~\cite{KoVa03}.
However, some care has to be exercised due to the fact that our characteristic matrix will not
be uniquely determined by the code (see Example~\ref{E-chitrellisnonunique} below).

%%%%%%%%%%%%%%%%%%%%%%%%%%%%
\begin{defi}\label{D-CharMat}
Let $\cC\subseteq\F^n$ be a $k$-dimensional code with support~$\cI$.
A {\sl characteristic pair\/} of~$\cC$ is defined to be a pair $(X,\cT)$, where
\begin{equation}\label{e-XY}
    X=\begin{pmatrix}x_1\\ \vdots\\ x_n\end{pmatrix}\in\F^{n\times n}\text{ and }
    \cT=\big[(a_l,b_l],\,l=1,\ldots,n\big]
\end{equation}
have the following properties
\begin{romanlist}
\item $\im X=\cC$, that is, $\{x_1,\ldots,x_n\}$ forms a generating set of~$\cC$.
\item $(a_l,\,b_l]$ is a span of~$x_l$ for $l=1,\ldots,n$.
\item $a_1,\ldots,a_n$ are distinct and $b_1,\ldots,b_n$ are distinct.
\item For all $j\in\cI$ there exist exactly $n-k$ row indices $l_1,\ldots,l_{n-k}$ such that
      $j\in(a_{l_i},\,b_{l_i}]$ for $i=1,\ldots,n-k$.
\end{romanlist}
We call~$X$ a {\sl characteristic matrix\/} of~$\cC$ and~$\cT$ the {\sl characteristic span list}.
The {\sl span matrix\/} of the characteristic pair $(X,\cT)$ is defined to be the matrix
\[
    S=(s_{lj})_{l=1,\ldots,n\hfill\atop j=0,\ldots,n-1}\in\N_0^{n\times n}, \text{ where }
    s_{lj}=\left\{\begin{array}{ll}1,&\text{if }j\in(a_l,\,b_l],\\[.3ex]0,&\text{if }j\not\in(a_l,\,b_l].
           \end{array}\right.
\]
\end{defi}
%%%%%%%%%%%%%%%%%%%%%%%%%%%

As we will see below, the characteristic span list is, up to ordering, uniquely determined by the code~$\cC$.
This is one of the results proven in~\cite{KoVa03}.
The only difference between our definition and the one of Koetter/Vardy is that in~\cite{KoVa03} for each span $(a_l,b_l]$
the corresponding row in~$X$ is the unique lexicographically first codeword having this span (with a slight modification
in the non-binary case).
As Koetter/Vardy point out in \cite[Remark, p.~2095]{KoVa03} there is no need for this particular choice.

Here are some first simple properties of characteristic pairs.

%%%%%%%%%%%%%%%%%%%%%%%%%%%
\begin{rem}\label{R-CharMat}
\begin{alphalist}
\item Every column of the span matrix~$S$ has exactly $n-k$ entries equal to~$1$ and~$k$ entries equal to~$0$.
      This is simply a reformulation of Definition~\ref{D-CharMat}(iv).
\item Let~$\sigma$ be the left shift on~$\F^n$ as in Remark~\ref{R-shiftKV}.
      If $(X,\cT)$ as in\eqnref{e-XY} is a characteristic pair for~$\cC$, then
      $\big(X^*,\cT^*\big)$ is a characteristic pair of~$\sigma(\cC)$, where
      \[
         X^*:=\begin{pmatrix}\sigma(x_1)\\ \vdots\\ \sigma(x_n)\end{pmatrix}
         \text{ and }
         \cT^*:=\big[(a_l-1,\, b_l-1],\,l=1,\ldots,n\big].
      \]
\item By part~(iv) of the definition there exist exactly~$k$ linear spans in the span list~$\cT$
      (recall that $(a,b]$ is linear iff $0\not\in(a,b]$).
      Therefore and due to~(iii), the corresponding rows of~$X$ form an MSGM of~$\cC$ in the sense of
      \cite[Def.~6.2, Thm.~6.11]{McE96}.
      Likewise, the previous item shows that for each $j=0,\ldots,n-1$ a characteristic matrix
      contains a shifted MSGM for the shifted code $\sigma^j(\cC)$.
      In this sense, the rows of a characteristic matrix may be called a set
      of ``shortest generators'' of~$\cC$ (and all its cyclic shifts) in the sense of \cite{Fo09}.
\end{alphalist}
\end{rem}
%%%%%%%%%%%%%%%%%%%%%%%%

So far we have not established the existence of characteristic pairs for a given code.
This crucial result, along with the uniqueness of the span set, has been proven in \cite{KoVa03}
and an algorithm has been given for computing a characteristic matrix.
All this leads to the following result.

%%%%%%%%%%%%%%%%%%%%%%%%%%
\begin{theo}\label{T-CharMat}
Let $\cC\subseteq\F^n$ be a $k$-dimensional code with support~$\cI$.
Then~$\cC$ has a characteristic pair and the characteristic span list is, up to ordering, uniquely determined by~$\cC$.
Thus, the rows of the span matrix~$S$ in Definition~\ref{D-CharMat} are uniquely determined up to ordering.
\end{theo}
%%%%%%%%%%%%%%%%%%%%%%%%%%%

\begin{proof}
The existence of a characteristic pair with the properties as in Definition~\ref{D-CharMat}(i)~-- (iv) has been shown in
\cite[Lem.~5.7, Thm.~5.9, Thm.~5.10 and its proof]{KoVa03}.
Since our notion of characteristic pair is slightly more general than the one in~\cite{KoVa03}, we need to pay special
attention to the uniqueness of the characteristic span list.
In order to do so let $(X,\cT)$  as in\eqnref{e-XY} be a characteristic pair of~$\cC$.
Fix $j\in\{0,\ldots,n-1\}$ and choose the~$k$ indices $l_1,\ldots,l_k$ for which $j\not\in(a_{l_i},b_{l_i}]$; they uniquely
exist due to property~(iv).
Then $0\not\in(a_{l_i}-j,b_{l_i}-j]$ and the latter is a span for the shifted row $\sigma^j(x_{l_i})$.
Hence those shifted rows have linear spans and, due to property~(iii) they form an
MSGM of~$\sigma^j(\cC)$ in the sense of \cite[Def.~6.2]{McE96}.
By \cite[Thm.~6.11]{McE96} the shifted span list is uniquely determined by the code~$\sigma^j(\cC)$ and hence by~$\cC$.
All this shows that the span list~$\cT$ is uniquely determined by~$\cC$ and coincides with the one of the
characteristic matrix as defined in~\cite[Def.~5.2]{KoVa03}.
\end{proof}

The matrix~$X$ of a characteristic pair is, in general, not uniquely determined by~$\cC$ (up to ordering of the rows), see
Example~\ref{E-chitrellisnonunique} below.

The importance of the characteristic pairs is that all minimal trellises of the code can be retrieved from them.
This is one of the main results in~\cite{KoVa03}.
Let us first introduce the following notion.

%%%%%%%%%%%%%%%%%%%%%%%%%%
\begin{defi}\label{D-KVtrellis}
Let $\cC\subseteq\F^n$ be a $k$-dimensional code with support~$\cI$ and let $(X,\cT)$ be a characteristic pair
of~$\cC$ as in\eqnref{e-XY}.
Any trellis of the form $T_{x_{l_1},(a_{l_1},\,b_{l_1}]}\times\ldots\times T_{x_{l_k},(a_{l_k},\,b_{l_k}]}$,
where $x_{l_1},\ldots,x_{l_k}$ are linearly independent rows of~$X$, is called a {\sl $\KV_{(X,\cT)}$-trellis\/} of~$\cC$.
Every trellis that is a $\KV_{(X,\cT)}$-trellis for some characteristic pair $(X,\cT)$ of $\cC$ is called a $\KV$-trellis of~$\cC$.
\end{defi}
%%%%%%%%%%%%%%%%%%%%%%%%%%

%%%%%%%%%%%%%%%%%%%%%%%%%%%
\begin{rem}\label{R-chi-trellis}
By Theorem~\ref{T-KVtrellis} every $\KV$-trellis is one-to-one.
Moreover, if $T=T_{x_{l_1},(a_{l_1},\,b_{l_1}]}\times\ldots\times T_{x_{l_k},(a_{l_k},\,b_{l_k}]}$ is a $\KV_{(X,\cT)}$-trellis, then
the SCP of~$T$ is given by $(v_1,\ldots,v_n)S$, where~$S$ is the span matrix of the characteristic pair~$(X,\cT)$ and
$v_i=1$ if $i\in\{l_1,\ldots,l_k\}$ and $v_i=0$ else.
This follows immediately from Theorem~\ref{T-KVtrellis}.
As a consequence, the SCP $(s_0,\ldots,s_{n-1})$ of a KV-trellis satisfies $s_i\leq\min\{k,\,n-k\}$ for all $i\in\cI$, which generalizes
the well-known Wolf bound~\cite{Wo78} (see also \cite[Thm.~5.5]{Va98}) to tail-biting trellises.
Finally, since a KV-trellis is a product trellis where the starting points of the spans are distinct and so are the end points,
the formulas for the ECP and SCP given in Proposition~\ref{P-SCPECP} apply.
\end{rem}
%%%%%%%%%%%%%%%%%%%%%%%%%%%%

Unfortunately, the set of $\KV_{(X,\cT)}$-trellises depends on the choice of the characteristic pair $(X,\cT)$.
Indeed, we have the following example.

%%%%%%%%%%%%%%%%%%%%%%%%%%
\begin{exa}\label{E-chitrellisnonunique}
Let $\F=\F_2$ and $\cC=\{(0000),(1100),(0111),(1011)\}\subseteq\F_2^4$.
The two pairs $(X,\cT)$ and $(X',\cT)$, where
\begin{equation}\label{e-charmat1}
   X=\begin{pmatrix}1&1&0&0\\0&1&1&1\\1&0&1&1\\0&1&1&1\end{pmatrix},\,
   X'=\begin{pmatrix}1&1&0&0\\0&1&1&1\\1&0&1&1\\1&0&1&1\end{pmatrix},\,
   \cT=\big[(0,1],\,(1,3],\,(2,0],\,(3,2]\big],
\end{equation}
are both characteristic pairs of~$\cC$.
Their common span matrix is
\[
   S=\begin{pmatrix}0&1&0&0\\0&0&1&1\\1&0&0&1\\1&1&1&0\end{pmatrix}.
\]
The last two rows of~$X$ are linearly independent and lead to a $\KV_{(X,\cT)}$-trellis with SCP $(2,1,1,1)$; see Remark~\ref{R-chi-trellis}.
But since the last two rows of~$X'$ are linearly dependent, this SCP does not appear for any $\KV_{(X',\cT)}$-trellis of~$\cC$
as one can see directly from the span matrix~$S$.
It is easy to see that the two matrices~$X$ and~$X'$ are the only two characteristic matrices for the code~$\cC$, up to ordering of the rows.
Indeed, examining the set of codewords it is clear that for the spans $(0,1],\,(1,3]$, and $(2,0]$ there is exactly
one codeword having this span, while the span $(3,2]$ is attained by the two codewords $(0111)$ and $(1011)$.
This leads exactly to the two choices of characteristic matrices given in\eqnref{e-charmat1}.
Since $(0111)$ is the lexicographically first of those two vectors, the matrix~$X$ is the characteristic matrix singled
out in~\cite{KoVa03}.
\end{exa}
%%%%%%%%%%%%%%%%%%%%%%%%%%

The following result expresses the fact that the list of spans uniquely determines the structure of KV-trellises.
For structural isomorphisms recall Definition~\ref{D-TBTbasics}(e).

%%%%%%%%%%%%%%%%%%%%%%%%%%%%%%%
\begin{cor}\label{C-KVtrellises}
Let~$T=T_{G,\cS}$ and~$T'=T_{G',\cS'}$ be KV-trellises of~$\cC$.
Then the following are equivalent:
\begin{romanlist}
\item $T$ and~$T'$ are structurally isomorphic.
\item $T$ and~$T'$ have the same SCP and ECP.
\item $\cS=\cS'$, up to ordering.
\end{romanlist}
\end{cor}
%%%%%%%%%%%%%%%%%%%%%%%%%%%%%%%

\begin{proof}
(i)~$\Rightarrow$~(ii) is obvious.
\\
(ii)~$\Rightarrow$~(iii) From Proposition~\ref{P-SCPECP}(a) it follows that the sets $\cA$ and $\cA'$ of
starting points of the spans in~$\cS$ and $\cS'$, respectively, are identical.
But for KV-trellises any starting point uniquely determines its span, see
Theorem~\ref{T-CharMat}. Hence $\cS=\cS'$, up to ordering.
\\
(iii)~$\Rightarrow$~(i)
Without loss of generality, let us assume $\cS=\cS'$. Then $M_i=M'_i$ for the product trellis matrices of~$T$ and~$T'$, respectively, defined
in\eqnref{e-MMat}.
But then the edge spaces of~$T$ and~$T'$ are given by $E_i=\im (M_i,G_i,M_{i+1})$ and
$E'_i=\im (M_i,G'_i,M_{i+1})$, where the $i$-th column of~$G$ and~$G'$ is denoted by $G_i$ and~$G'_i$, respectively.
By Proposition~\ref{P-SCPECP} we have $\dim E_i=\dim E'_i$.
This shows that $\phi_i\big(\alpha(M_i,G_i,M_{i+1})\big)=\alpha(M_i,G'_i,M_{i+1})$ is a well-defined
structural isomorphism between~$T$ and~$T'$.
\end{proof}

It is worth observing that the implication (ii)~$\Rightarrow$~(iii) is not true in general for product trellises.
For instance, the two matrices
\[
   G=\begin{pmatrix}1&1&0&1&0\\1&0&0&1&1\\1&1&1&0&1\end{pmatrix},\ G'=\begin{pmatrix}0&1&0&0&1\\0&0&1&1&1\\1&1&1&0&1\end{pmatrix}
   \in\F_2^{3\times5}
\]
generate the same code of dimension~$3$.
Choosing the span list $\cS=\big[(3,1],(0,4],(4,2]\big]$ for~$G$ and $\cS'=\big[(4,1],(3,2],(0,4]\big]$ for~$G'$
results in product trellises $T_{G,\cS}$ and $T_{G',\cS'}$ with the same SCP $(2,3,2,1,2)$ and ECP $(3,3,2,2,3)$ even though $\cS\not=\cS'$.
However, one can check that the two trellises are not structurally isomorphic.
As a consequence, the question whether~(i) above is equivalent to~(iii) remains open for (one-to-one) product trellises.

Let us now state the main result of~\cite{KoVa03}.

%%%%%%%%%%%%%%%%%%%%%%%%
\begin{theo}[\mbox{\cite[Thm.~5.5]{KoVa03}}] \label{T-min-trellis}
Let $\cC\subseteq\F^n$ be a $k$-dimensional code with support~$\cI$ and
let~$T$ be a minimal trellis of~$\cC$ in the sense of Definition~\ref{D-TBTmin}(a).
Then $T$ is one-to-one and there exists a characteristic pair $(X,\cT)$ such that~$T$ is a $\KV_{(X,\cT)}$-trellis.
\end{theo}
%%%%%%%%%%%%%%%%%%%%%%%%

\begin{proof}
Since our result is phrased differently from what is in \cite[Thm.~5.5]{KoVa03} we think it is worthwhile to give a proof.
First of all,~$T$ is a product trellis by \cite[Thm.~4.2]{KoVa03} and is one-to-one due to
\cite[Lemma~5.2]{KoVa03}.
Hence~$T$ is a product trellis as in our Theorem~\ref{T-KVtrellis} because, due to minimality, the trellis is non-mergeable and
thus we may assume shortest spans for all generators; see also the paragraphs following Theorem~\ref{T-KVtrellis}.
Thus, $T=T_{x_1,(a_1,b_1]}\times\ldots\times T_{x_k,(a_k,b_k]}$ for certain $x_l\in\cC$ with spans $(a_l,b_l]$.
Now \cite[Lemma~5.4]{KoVa03} implies that $\big(x_1,(a_1,b_1]\big),\ldots,(x_k,(a_k,b_k]\big)$ are characteristic
generators in the sense of \cite[Def~5.2]{KoVa03}.
But this simply means that this list can be extended to a characteristic pair $(X,\cT)$ of~$\cC$ in the sense of our
Definition~\ref{D-CharMat} and thus $T$ is a $\KV_{(X,\cT)}$-trellis.
\end{proof}

One should bear in mind that not all KV-trellises are minimal.
For instance, the last two rows of~$X$ in Example~\ref{E-chitrellisnonunique} are linearly independent
and give rise to a trellis with SCP $(2,1,1,1)$, but the first two rows of~$X$ are also linearly independent
and lead to a trellis with SCP $(0,1,1,1)$.
This shows that the first trellis is not minimal in the sense of Definition~\ref{D-TBTmin}(a).
A larger example is given in \cite[Exa.~6.1]{KoVa03}.

It is important to notice that the above result does not imply that all minimal trellises of~$\cC$ arise from
a fixed characteristic pair.
Indeed, we have the following example.

%%%%%%%%%%%%%%%%%%%%%%%%%%%%%
\begin{exa}\label{E-CMmintrellis}
Let $\F=\F_2$ and $\cC:=\im G$, where
\[
    G=\begin{pmatrix}0&1&0&0&0&1\\0&0&1&1&1&0\\1&0&1&0&1&0\end{pmatrix}\in\F^{3\times6}.
\]
Put
\[
   X=\begin{pmatrix}1&0&0&1&0&0\\0&1&0&0&0&1\\0&0&1&1&1&0\\1&0&0&1&0&0\\0&1&0&0&0&1\\1&0&1&0&1&0\end{pmatrix},\quad
   X'=\begin{pmatrix}1&0&0&1&0&0\\0&1&0&0&0&1\\0&0&1&1&1&0\\1&0&0&1&0&0\\0&1&0&0&0&1\\1&1&1&0&1&1\end{pmatrix}
\]
and $\cT=\big[(0,3],(1,5],(2,4],(3,0],(5,1],(4,2]\big]$.
It is easy to verify that $(X,\cT)$ and $(X',\cT)$ are both characteristic pairs of~$\cC$.
Put $\cS=\big[(1,5],(2,4],(4,2]\big]$.
Then the product trellises $T=T_{G,\cS}$ and $T'=T_{G',\cS}$, where
\[
   G'=\begin{pmatrix}0&1&0&0&0&1\\0&0&1&1&1&0\\1&1&1&0&1&1\end{pmatrix}
\]
are both KV-trellises based on the same choice of characteristic spans; the first one is a $\KV_{(X,\cT)}$-trellis
and the second one is the corresponding $\KV_{(X',\cT)}$-trellis.
The KV-trellises are minimal as one can check by going through the SCP's of all possible $\KV_{(X,\cT)}$-trellises.
(There are~$10$ choices for selecting~$3$ linearly independent rows in~$X$ and none of them leads to a smaller trellis.)
As one can verify directly from the graphs below, the trellises~$T$ and~$T'$ are structurally isomorphic, but not isomorphic.
\begin{center}
  \includegraphics[height=3.8cm]{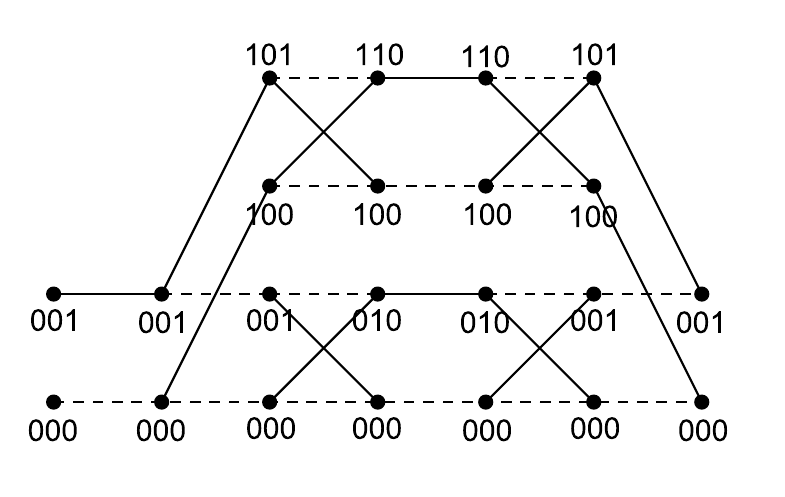}\quad\includegraphics[height=3.8cm]{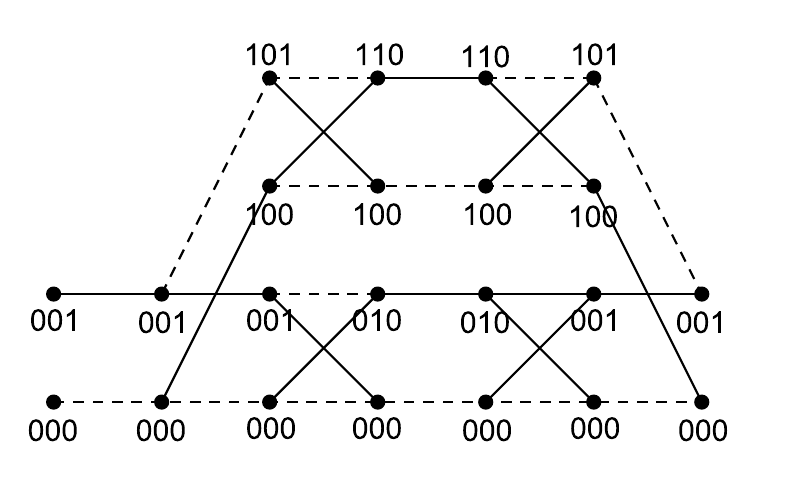}
  \\
  \hspace*{.8cm}\tiny{[$T_{G,\cS}$]\hspace*{5.3cm}[$T_{G',\cS}$]}
\end{center}
\end{exa}
%%%%%%%%%%%%%%%%%%%%%%%%%%%%%%

In light of the previous example it should be observed that, at this point, the following question is open.
Given a code $\cC\subseteq\F^n$ and fix a characteristic pair $(X,\cT)$.
Is then each minimal trellis of~$\cC$ structurally isomorphic to a $\KV_{(X,\cT)}$-trellis?
From Corollary~\ref{C-KVtrellises} we know that the SCP and ECP of the given trellis uniquely specify
the row selection of~$X$ that may result in a potentially structurally isomorphic $\KV_{(X,\cT)}$-trellis~$T'$.
But it is not clear whether this row selection is linearly independent.
Only then the trellis~$T'$ represents~$\cC$ and thus is a $\KV_{(X,\cT)}$-trellis.

%%%%%%%%%%%%%%%%%%%%%%%%%%
\begin{rem}\label{R-min-trellis}
As shown in \cite[Thm.~5.6]{KoVa03}, the same result as in Theorem~\ref{T-min-trellis} applies to many other notions
of minimality for tail-biting trellises, see~\cite[p.~2085]{KoVa03}.
\end{rem}
%%%%%%%%%%%%%%%%%%%%%%%%%%

%%%%%%%%%%%%%%%%%%%%%%%%%%%%%%%%%%%%%
\Section{The Tail-Biting BCJR-Construction}\label{S-BCJR}
%%%%%%%%%%%%%%%%%%%%%%%%%%%%%%%%%%%%%
In this section we will present a different way of producing tail-biting trellises for a given code.
These BCJR-constructions were introduced in~\cite{NoSh06}.
We will show that the resulting trellises are non-mergeable and that product trellises can always be merged into a corresponding
BCJR-trellis by taking suitable quotients.
We will see that KV-trellises are isomorphic to their BCJR-counterpart, and, consequently, non-mergeable.
Throughout let~$\cC$ be as in\eqnref{e-Gdata} and\eqnref{e-Hdata}.

%%%%%%%%%%%%%%%%%%%%%%%%%%%%%%%%%
\begin{theo}[\mbox{\cite[Lemma~2]{NoSh06}}]\label{T-NS}
Let $D\in\F^{k\times(n-k)}$ be any matrix and for $i\in\cI$ define the matrices
\begin{equation}\label{e-NMat}
  N_i=\left\{\begin{array}{ll}D,&\text{if }i=0\\[.3ex]
                              N_{i-1}+G_{i-1} H_{i-1},&\text{if }i>0
     \end{array}\right.
\end{equation}
Then $N_n=N_0$ and the trellis~$T:=T_{(G,H,D)}$ with vertex spaces $V_i:=\im N_i\subseteq\F^{n-k}$ and edge spaces
$E_i=\im (N_i,G_i,N_{i+1})$ is a linear, reduced and biproper trellis representing the code~$\cC$, that is, $\cC(T)=\cC$.
We call~$D$ the displacement matrix for the trellis~$T$.
\end{theo}
%%%%%%%%%%%%%%%%%%%%%%%%%%%%%%%%%

The fact $N_n=N_0$ simply follows from the identity $N_i=N_0+\sum_{j=0}^{i-1}G_jH_j,\,i=1,\ldots,n$ along with
$0=G H\T=\sum_{j=0}^{n-1}G_j H_j$.
The reducedness and biproperness are a consequence of the very definition of the vertex and edge spaces along with the
recursive definition of the matrices~$N_i$.
The displacement matrix~$D$ can be interpreted as the (circular) past at time zero of each cycle in~$T_{(G,H,D)}$.
If $D=0$, each cycle starts with zero history and the trellis~$T_{(G,H,D)}$ is conventional.
In fact, it is simply the classical BCJR-trellis and thus a minimal trellis, see \cite[Sec.~IV]{McE96} and~\cite{BCJR74}.
This case will also show up as a special version of Definition~\ref{D-NSspan} below, in which we define what will be
called a tail-biting BCJR-trellis.
One should regard the displacement matrix~$D$ as a design parameter for the construction of good trellises.
In Definition~\ref{D-NSspan} we will restrict ourselves to particularly useful choices of this parameter, which then will
always lead to non-mergeable trellises.

%%%%%%%%%%%%%%%%%%%%%%%%%%%%%%%%
\begin{rem}\label{R-NS}
For every $U\in GL_k(\F)$ the trellises $T_{(G,H,D)}$ and $T_{(UG,H,UD)}$ are identical and for all $V\in GL_{n-k}(\F)$
the trellises $T_{(G,H,D)}$ and $T_{(G,VH,DV^{\sf T})}$ are isomorphic and the isomorphisms on the vertex spaces are given by
the restrictions of the linear map~$V\T$ to those spaces.
\end{rem}
%%%%%%%%%%%%%%%%%%%%%%%%%

The statement $\cC(T)=\cC$ of Theorem~\ref{T-NS} also follows from the following result which will be helpful later on.
%%%%%%%%%%%%%%%%%%%%%%%%%%%%%%%%%
\begin{prop}\label{P-edgepath}
Let $T=T_{(G,H,D)}$ be as in Theorem~\ref{T-NS}.
Furthermore, let $\alpha^{(i)}\in\F^k$ be such that the sequence of edges
$(\alpha^{(i)} N_i,\alpha^{(i)}G_i,\alpha^{(i)} N_{i+1}),\ i\in\cI$,
is a path through~$T$, that is,
$\alpha^{(i)} N_{i+1}=\alpha^{(i+1)}N_{i+1}$ for $i=0,\ldots,n-2$.
Then the edge-label sequence of the path is a codeword if and only if the path is a cycle.
In other words,
\[
  (\alpha^{(0)}G_0,\ldots,\alpha^{(n-1)}G_{n-1})\in\cC\Longleftrightarrow \alpha^{(n-1)}N_0=\alpha^{(0)}N_0.
\]
\end{prop}
%%%%%%%%%%%%%%%%%%%%%%%%%%%%%%%%
\begin{proof}
From the assumptions as well as the definition of the matrices~$N_i$ we obtain
$\alpha^{(1)}N_1=\alpha^{(0)}N_1=\alpha^{(0)}N_0+\alpha^{(0)}G_0H_0$ and inductively
$\alpha^{(i+1)}N_{i+1}=\alpha^{(0)} N_0+\sum_{j=0}^{i}\alpha^{(j)}G_jH_j$, for $i=0,\ldots,n-2$.
Hence
$\alpha^{(n-1)}N_{n-1}=\alpha^{(0)} N_0+\sum_{j=0}^{n-2}\alpha^{(j)}G_jH_j$ and this implies
\[
  \alpha^{(n-1)}N_0=\alpha^{(n-1)}N_{n-1}+\alpha^{(n-1)}G_{n-1}H_{n-1}=\alpha^{(0)} N_0+\sum_{j=0}^{n-1}\alpha^{(j)}G_jH_j
\]
This leads to
$\alpha^{(n-1)} N_0=\alpha^{(0)}N_0\Longleftrightarrow\sum_{i=0}^{n-1}\alpha^{(i)}G_iH_i=0$ and the latter is the case if and only if
$(\alpha^{(0)}G_0,\ldots,\alpha^{(n-1)}G_{n-1})\in\ker H^{\sf T}=\cC$,
as desired.
\end{proof}

One of the nice properties of the trellis construction in Theorem~\ref{T-NS} is that it respects duality.
Indeed, the following is straightforward; see also \cite[Lemma~12]{NoSh06}.

%%%%%%%%%%%%%%%%%%%%%%%%
\begin{prop}\label{P-NSdual}
Let $T=T_{(G,H,D)}$ be as in Theorem~\ref{T-NS}.
Then $\hat{T}:=T_{(H,G,D^{\sf T})}$ is a trellis representing the dual code~$\cC^{\perp}$.
Moreover, the vertex spaces of $\hat{T}$ are given by $\im N_i\T,\,i\in\cI$, and thus the
trellises~$T$ and~$\hat{T}$ have the same SCP, given by $(s_0,\ldots,s_{n-1})$, where $s_i=\rk N_i$.
\\
As a consequence, if~$T$ is a minimal trellis for~$\cC$, then~$\hat{T}$ is a minimal trellis for~$\cC^{\perp}$.
\end{prop}
%%%%%%%%%%%%%%%%%%%%%%%

Later on in the proof of Theorem~\ref{T-dualselect} we will see that in the special case where
$T_{(G,H,D)}$ is a minimal KV-trellis, this trellis dualization coincides with the construction given by Forney in
\cite[Sec.~VII-D and Thm.~8.4]{Fo01} based on dualizing the edge spaces (local constraints) with respect
to a particular inner product.

A particular instance of the trellis $T_{(G,H,D)}$ is obtained by choosing the displacement matrix based on given
spans of the generator matrix~$G$.
This will relate this construction to the product construction of the previous section.

%%%%%%%%%%%%%%%%%%%%%%%%%%%
\begin{defi}\label{D-NSspan}
Let~$\cS:=\big[(a_l,b_l],l=1,\ldots,k\big]$ be a span list of~$G$.
Then the trellis $T_{(G,H,\cS)}$ defined as $T_{(G,H,D)}$, where
\begin{equation}\label{e-DMat}
   D=\begin{pmatrix}d_1\\ \vdots\\ d_k\end{pmatrix}\in\F^{k\times(n-k)}\text{ with }d_l=\sum_{j=a_l}^{n-1}g_{lj} H_j,
\end{equation}
is called a {\sl (tail-biting) BCJR-trellis\/} of~$\cC$.
\end{defi}
%%%%%%%%%%%%%%%%%%%%%%%%%%%%

It should be noted that we define a BCJR-trellis more restrictive than Nori/Shankar in~\cite{NoSh06}:
while we use this terminology only for trellises $T_{(G,H,D)}$, where the displacement matrix is based on a given span list
as in\eqnref{e-DMat}, Nori/Shankar use the term BCJR-trellis for all displacement matrices.

%%%%%%%%%%%%%%%%%%%%%%%%%%%%
\begin{rem}\label{R-BCJR}
From Remark~\ref{R-NS} it follows that $T_{(G,H,\cS)}$ is isomorphic to $T_{(G,VH,\cS)}$ for all $V\in GL_{n-k}(\F)$.
Indeed, if~$D$ is as in\eqnref{e-DMat} then~$DV\T$ is the according displacement matrix for $T_{(G,VH,\cS)}$.
As a consequence,~$V\T$ restricted to all vertex spaces furnishes an isomorphism between the two trellises.
However, if~$G$ and $UG\in\F^{k\times n},\,U\in GL_k(\F)$, are two generator matrices with the same span
list~$\cS$, then the BCJR-trellises $T_{(G,H,\cS)}$ and $T_{(UG,H,\cS)}$ need not be isomorphic.
This can be seen from Example~\ref{E-CMmintrellis}.
Indeed, the two (non-isomorphic) trellises given in that example are
BCJR-trellises, which will follow from Theorem~\ref{T-KVNS} presented below.
\end{rem}
%%%%%%%%%%%%%%%%%%%%%%%%%%%%%

Notice that if all spans are linear, then~$D$ is the zero matrix (see also Proposition~\ref{P-NSspan} below), and the resulting
trellis is the conventional BCJR-trellis.
It is well-known that the conventional BCJR-trellis is minimal, see \cite[Sec.~IV]{McE96}, and thus non-mergeable.
Later on we will see that in the tail-biting case, BCJR-trellises are non-mergeable as well, but they are not necessarily minimal.

Just like the product trellises in Remark~\ref{R-shiftKV} the BCJR-construction is compatible with the shift.
%%%%%%%%%%%%%%%%%%%%%%%%%%%%
\begin{rem}\label{R-shiftNS}
Let~$\cS$ and~$D$ be as in Definition~\ref{D-NSspan}.
As in Remark~\ref{R-shiftKV} let $G^*\in\F^{k\times n}$ be the matrix consisting of the shifted rows
$\sigma(g_l),\,l=1,\ldots,k$ and define~$H^*$ accordingly.
Then $\cS^*:=\big[(a_l-1,\, b_l-1],\,l=1,\ldots,k\big]$ is a span list for the rows of~$G^*$.
Let $N_i,\,i\in\cI$, be the vertex matrices as in\eqnref{e-NMat} for the BCJR-trellis $T_{(G,H,\cS)}$.
Then the vertex matrices for the BCJR-trellis $T_{(G^*,H^*,\cS^*)}$ are given by $N_i^*=N_{i+1}$ for $i\in\cI$.
Indeed, by\eqnref{e-DMat} the $l$-th row of $N^*_0$ is defined as
\[
  \sum_{j=a_l-1}^{n-1}g^*_{lj}H^*_j=\sum_{j=a_l-1}^{n-1} g_{l,j+1}H_{j+1}
  =\sum_{j=a_l}^{n-1}g_{lj}H_j+g_{l0}H_0
\]
and the latter is the $l$-th row of~$N_1$.
Now the rest follows inductively due to\eqnref{e-NMat}.
\end{rem}
%%%%%%%%%%%%%%%%%%%%%%%%%%%%

In the sequel we will show that if the pair~$(G,\cS)$ originates from a characteristic pair of~$\cC$, then
the trellis $T_{(G,H,\cS)}$ is isomorphic to the corresponding $\KV$-trellis $T_{G,\cS}$.

We begin by showing the following.
%%%%%%%%%%%%%%%%%%%%%%%%%%%%%%%%
\begin{prop}\label{P-NSspan}
Let $T_{(G,H,\cS)}$ be as in Definition~\ref{D-NSspan} and let $N_i,\,i\in\cI,$ be defined as in Theorem~\ref{T-NS}.
Then
\begin{alphalist}
\item If $(a_l, b_l]$ is a linear span, then $d_l=0$.
\item If $i\in\cI$ and $l\in\{1,\ldots,k\}$ are such that
      $i\not\in (a_l,\,b_l]$, then the $l$-th row of $N_i$ is zero.
\end{alphalist}
\end{prop}
%%%%%%%%%%%%%%%%%%%%%%%%%%%%%%%
\begin{proof}
(a) If $(a_l,b_l]$ is linear then all nonzero entries of~$g_l$ are in the interval $[a_l,n-1]$ and therefore
$d_l=\sum_{j=a_l}^{n-1}g_{lj} H_j=\sum_{j=0}^{n-1}g_{lj} H_j=0$ due to $G H\T=0$.
\\
(b) Recall that $N_0=D$ as well as the fact that $(a_l, b_l]$ is a linear span iff $0\not\in (a_l,\,b_l]$.
Therefore, the case $i=0$ has been proven in~(a).
Thus, let $i\geq1$.
Then, by definition, $N_i=D+\sum_{j=0}^{i-1}G_j H_j$.
\\
Suppose first that the span $(a_l,b_l]$ is circular, hence in particular, non-empty.
Then $i\not\in (a_l,\,b_l]$ is equivalent to $i\in(b_l,a_l]$, which is linear; see also\eqnref{e-compint}.
Thus, the $l$-th row of~$N_i$ is
\[
  d_l+\sum_{j=0}^{i-1}g_{lj} H_j=\sum_{j=a_l}^{n-1}g_{lj} H_j+\sum_{j=0}^{i-1}g_{lj} H_j=
  \sum_{j=0}^{n-1}g_{lj} H_j=0,
\]
where the second identity holds true because $g_{lj}=0$ for $j\in [i,a_l-1]$.
\\
Let now the span $(a_l,b_l]$ be linear. In this case $d_l=0$ by part~(a) and the $l$-th row of~$N_i$ is
given by
\begin{equation}\label{e-rowMi}
  \sum_{j=0}^{i-1}g_{lj} H_j.
\end{equation}
Now $i\not\in (a_l,\,b_l]$ amounts to either $i\in[0,a_l]$ or $i\in(b_l,n-1]$.
In the first case, the interval $[0,i-1]$ and the closed span $[a_l,\,b_l]$ of~$g_l$ are disjoint while in the second case
the interval $[0,i-1]$ contains $[a_l,b_l]$.
Hence in both cases the sum in\eqnref{e-rowMi} is zero.
\end{proof}

Now we are in a position to prove the non-mergeability of the BCJR-trellises.
This result generalizes the well-known fact that the conventional BCJR-trellis is minimal (thus non-mergeable),
see~\cite[Sec.~V]{McE96}.
It is worth noticing that, just like in the conventional case, the corollary is true even without
assuming that the starting points $a_1,\ldots,a_k$ of the spans (or the endpoints $b_1,\ldots,b_k$) are distinct.

%%%%%%%%%%%%%%%%%%%%%%%%%%%%
\begin{cor}\label{C-NSnonmerg}
Any BCJR-trellis $T:=T_{(G,H,\cS)}$ is non-mergeable.
\end{cor}
%%%%%%%%%%%%%%%%%%%%%%%%%%%%

\begin{proof}
Let~$V_i$ be the $i$-th vertex space of $T_{(G,H,\cS)}$.
Let us first suppose that the vertices~$v,\,v'\in V_0$ are mergeable.
Then the edge-label sequence of each path through~$T$ starting at~$v\in V_0$ and ending in $v'\in V_0$ must be a
codeword in~$\cC$ (because after merging this path would be a cycle).
But then by linearity of the trellis the edge-label sequence of each path through~$T$ starting at~$v-v'\in V_0$ and
ending in $0\in V_0$ must be a codeword in~$\cC$ as well.
Due to Proposition~\ref{P-NSspan}(b) we may apply Lemma~\ref{L-pathv0} and conclude that such a path does indeed exist in~$T$.
But then Proposition~\ref{P-edgepath} implies that this path is a cycle. Hence $v-v'=0$.
All this shows that no vertices in~$V_0$ are mergeable.
Now the non-mergeability of vertices in any~$V_i$ follows from the shift property of~$T$ as described in
Remark~\ref{R-shiftNS}.
\end{proof}

It should be noted that if the displacement matrix~$D$ is not as in Definition~\ref{D-NSspan}, then
the trellises $T_{(G,H,D)}$ as in Theorem~\ref{T-NS} are, in general, mergeable even if the starting points (resp., end points)
of the spans are distinct.
An example is given in \cite[Exa.~9, Fig.~14]{NoSh06}.
One can easily see how to merge that trellis.
Moreover, that trellis has SCP $(2,2,2,2)$ and therefore Proposition~\ref{P-NSspan}(b) implies that the underlying displacement
matrix is not as in Definition~\ref{D-NSspan} and thus the trellis is not a BCJR-trellis in our sense.
Furthermore, that trellis does not contain any paths connecting a nonzero vertex in~$V_0$ with
$0\in V_0$ and therefore the main argument of the proof above does not apply.

At this point it should be mentioned that BCJR-trellises are in general not one-to-one, not even if the starting points
(resp.~endpoints) of the spans are distinct.
This is illustrated by the following example taken from \cite[Exa.~8]{NoSh06}.

%%%%%%%%%%%%%%%%%%%%%%%%%%%%
\begin{exa}\label{E-BCJRnotone}
Let $\cC=\im G=\ker H\T$, where
\[
    G=\begin{pmatrix}1&0&1\\1&1&0\end{pmatrix}\in\F_2^{2\times3}, \text{ and }
    \cS=\big[(0,2],\,(1,0]\big], \ H=\begin{pmatrix}1&1&1\end{pmatrix}.
\]
The BCJR-trellis $T:=T_{(G,H,\cS)}$ has edge spaces $E_i=\im(N_i,G_i,N_{i+1})$ given by
\[
  E_0=\im\!\!\left(\!\!\begin{array}{c|c|c}
    0&1&1\\1&1&0\end{array}\!\!\right),\
  E_1=\im\!\!\left(\!\!\begin{array}{c|c|c}
    1&0&1\\0&1&1\end{array}\!\!\right),\
  E_2=\im\!\!\left(\!\!\begin{array}{c|c|c}
    1&1&0\\1&0&1\end{array}\!\!\right).
\]
This shows that $1\edge{0}1\edge{0}1\edge{0}1$ is a non-trivial cycle in~$T$ with the zero codeword as
edge-label sequence; see also the figure below.
Hence~$T$ is not one-to-one.
Let us also consider the associated product trellis $\hat{T}=T_{G,\cS}$.
In this case the edge spaces are given by
\[
  \hat{E}_0=\im\!\!\left(\!\!\begin{array}{cc|c|cc}
    0&0&1&1&0\\0&1&1&0&0\end{array}\!\!\right),\
  \hat{E}_1=\im\!\!\left(\!\!\begin{array}{cc|c|cc}
    1&0&0&1&0\\0&0&1&0&1\end{array}\!\!\right),\
  \hat{E}_2=\im\!\!\left(\!\!\begin{array}{cc|c|cc}
    1&0&1&0&0\\0&1&0&0&1\end{array}\!\!\right).
\]
From Theorem~\ref{T-KVtrellis} we know that~$\hat{T}$ is one-to-one; see figure below.
However, the trellis~$\hat{T}$ is mergeable at time $i=2$ and indeed, merging the states $(11),\,(00)$ and
the states $(10),\,(01)$ (that is, replacing~$V_2$ by the quotient space $V_2/\im(11)$)
leads to the (not one-to-one) trellis~$T=T_{(G,H,\cS)}$.
In the following graphs as well as in all later graphs a solid line denotes an edge with label~$1$, while dashed lines denote the label~$0$.
\\
\mbox{}\hspace*{1.5cm}\includegraphics[height=3.7cm]{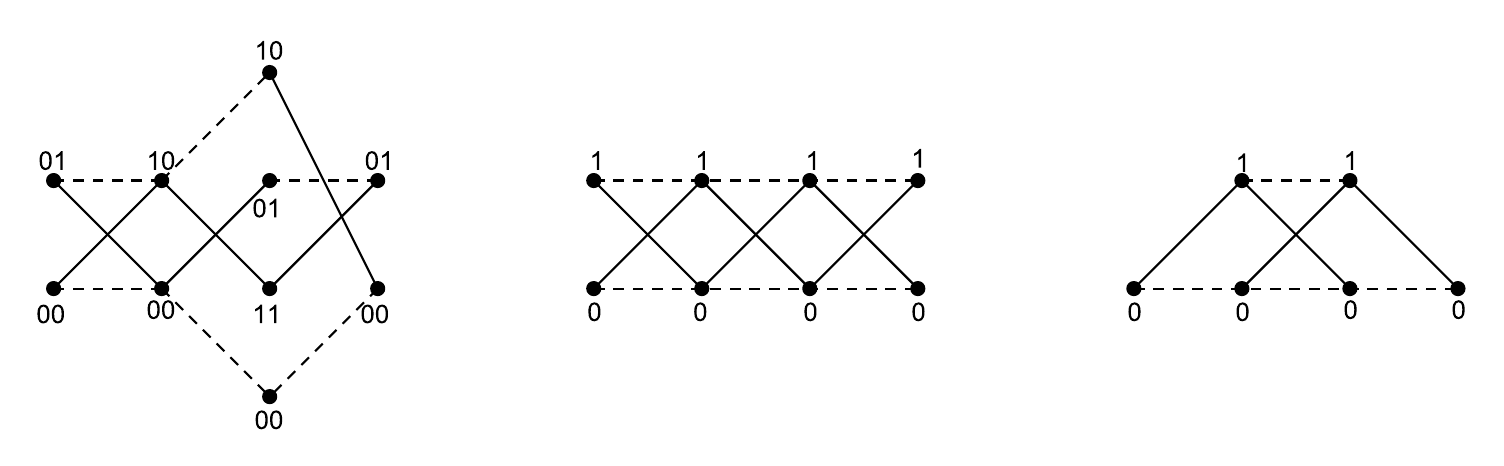}
\\[-1ex]
\mbox{}\hspace*{3cm}{\tiny [Trellis~$\hat{T}$]} \hspace*{2.5cm} {\tiny [Trellis~$T$]}\hspace*{2.2cm}{\tiny [Conventional Minimal Trellis]}

The trellis~$T$, being a BCJR-trellis, is non-mergeable.
However, it is not minimal as the conventional minimal trellis of~$\cC$ at the very right shows.
It should also be noted that~$\hat{T}$ is not (isomorphic to) a KV-trellis.
Indeed, as one can easily verify, the unique characteristic pair of~$\cC$ with span matrix~$S$ is given by
\[
   (X,\,\cT)=\Big(\begin{pmatrix}0&1&1\\1&1&0\\1&0&1\end{pmatrix},\,\big[(1,2],\,(0,1],\,(2,0]\big]\Big),\
   S=\begin{pmatrix}0&0&1\\0&1&0\\1&0&0\end{pmatrix}.
\]
From this it is obvious that the SCP $(1,1,2)$ of~$\hat{T}$ cannot be attained by a KV-trellis and neither can the
SCP $(1,1,1)$ of the trellis~$T$.
Moreover, the trellises~$\hat{T}$ and~$T$ cannot be merged into a KV-trellis of the code.
\\
On the other hand, in Theorem~\ref{T-mergeProd} we will see that the process of merging a product trellis~$\hat{T}$ into the corresponding
BCJR-trellis~$T$ is always possible (if~$\hat{T}$ is mergeable at all) and can be achieved by taking a suitable quotient trellis.
\end{exa}
%%%%%%%%%%%%%%%%%%%%%%%%%%%%%

%%%%%%%%%%%%%%%%%%%%%%%%%%%
\begin{theo}\label{T-isomorphic}
Let $\cS:=\big[(a_l,b_l],\,l=1,\ldots,k\big]$ and $T:=T_{(G,H,\cS)}$ be as in Definition~\ref{D-NSspan}.
Let~$\hat{T}$ be the product trellis $\hat{T}:=T_{G,\cS}$.
Then the vertex spaces $V_i:=\im N_i$ and $\hat{V}_i:=\im M_i,\,i\in\cI,$ of~$T$ and~$\hat{T}$, respectively,
satisfy
\begin{equation}\label{e-VVhat}
  \dim V_i\leq \dim \hat{V}_i\text{ for }i\in\cI.
\end{equation}
If equality holds true for all~$i\in\cI$ in\eqnref{e-VVhat}, that is,~$T$ and~$\hat{T}$ have the same SCP, then
one has the following.
\begin{alphalist}
\item $T$ and~$\hat{T}$ are isomorphic, thus the product trellis~$\hat{T}$ is non-mergeable.
\item $a_1,\ldots,a_k$ are distinct and so are $b_1,\ldots,b_k$. Therefore, the formulas for the ECP and SCP given
in Proposition~\ref{P-SCPECP} apply.
\end{alphalist}
\end{theo}
%%%%%%%%%%%%%%%%%%%%%%%%%%%%
It is worth noting that $T\preceq_{\Theta}\hat{T}$ with respect to the $\Theta$-ordering defined in \cite[p.~2084]{KoVa03}.
While we have seen in Proposition~\ref{P-SCPECP} that if both $a_1,\ldots,a_k$ and $b_1,\ldots,b_k$  are distinct,
the formulas for the ECP do always apply to product trellises, this is not the case for the BCJR-trellis;
see for instance Example~\ref{E-BCJRnotone}.

\begin{proof}
First of all, by Proposition~\ref{P-NSspan}(b) and the very definition of the matrices~$M_i$ in\eqnref{e-MMat}
we have
\begin{equation}\label{e-colNM}
   \col(N_i)\subseteq\col(M_i),
\end{equation}
where $\col(M):=\{Mw\mid w\in \F^r\}$ denotes the column space of the matrix $M\in\F^{k\times r}$.
Hence $\rk(N_i)\leq\rk(M_i)$ and this proves\eqnref{e-VVhat}.
\\
Let us now assume that we have equality in\eqnref{e-VVhat} for all $i\in\cI$.
\\
(a) The inclusions\eqnref{e-colNM} imply
\begin{equation}\label{e-colNMeq}
   \col(N_i)=\col(M_i) \text{ for }i\in\cI
\end{equation}
and hence $\ker N_i=\ker M_i$, where, as usual the
kernel denotes the left kernel and thus is simply the orthogonal of the column space with respect to the standard inner product on~$\F^k$.
As a consequence, the map $\phi_i: \hat{V}_i\longrightarrow V_i,\quad \alpha M_i\longmapsto \alpha N_i$
is a well-defined isomorphism.
But then the maps~$\phi_i$ give rise to an isomorphism between~$\hat{T}$ and~$T$ because
the edge spaces of~$\hat{T}$ are given by $\hat{E}_i=\im(N_i,G_i, N_{i+1})$
while those of~$T$ are $E_i=\im(M_i,G_i, M_{i+1})$.
The non-mergeability of~$\hat{T}$ follows from Corollary~\ref{C-NSnonmerg}.
\\
(b) Let the index sets~$\cL_i$ be as in\eqnref{e-Li} again. Then $\rk(M_i)=|\cL_i|$ by Proposition~\ref{P-SCPECP}.
Since by assumption $\rk(N_i)=\rk(M_i)$ Proposition~\ref{P-NSspan}(b) implies
\begin{equation}\label{e-linind}
  \text{the rows of~$N_i$ with indices from~$\cL_i$ are linearly independent.}
\end{equation}
Assume first that for some $l\not=m$ we have $a_l=a_m$.
Put $i=a_l+1=a_m+1$.
Then the $l$-th row of $N_i$ is given by
\[
   \sum_{j=a_l}^{n-1}g_{lj}H_j+\sum_{j=0}^{i-1}g_{lj}H_j=g_{la_l}H_{a_l}
\]
and likewise the $m$-th row of $N_i$ is
\begin{equation}\label{e-Nim}
  \sum_{j=a_m}^{n-1}g_{mj}H_j+\sum_{j=0}^{i-1}g_{mj}H_j=g_{ma_l}H_{a_l}.
\end{equation}
Hence these two rows are linearly dependent.
Observe that if $a_l\not=b_l$ and $a_m\not=b_m$ then $l,\,m\in\cL_i$ and therefore this linear dependence contradicts\eqnref{e-linind}.
If $a_l=b_l=a_m\not=b_m$, then $(a_l,b_l]=\emptyset$ and $m\in\cL_i$, but $l\not\in\cL_i$. However, in this case
the row~$g_l$ is nonzero only at the $a_l$-th position and this forces $H_{a_l}=0$. Again,\eqnref{e-Nim} contradicts\eqnref{e-linind}.
The case $b_l\not=a_l=a_m=b_m$ is symmetric and the case $a_l=b_l=a_m=b_m$ cannot occur due to the linear independence
of the rows $g_l,\,g_m$.
\\
Finally assume that for some $l\not=m$ we have $b_l=b_m$.
Then for $i=b_l=b_m$ the $l$-th row of~$N_i$ is
$\sum_{j=a_l}^{n-1}g_{lj}H_j+\sum_{j=0}^{i-1}g_{lj}H_j=-g_{lb_l}H_{b_l}$
and likewise the $m$-th row of $N_i$ is $-g_{mb_l}H_{b_l}$. Now we may argue as in the first case.
\end{proof}

In the proof we used the fact that if equality holds in\eqnref{e-VVhat} then we have\eqnref{e-linind}.
With the aid of Proposition~\ref{P-NSspan}(b) this yields
\begin{equation}\label{e-kerNi}
  \ker N_i=\{\alpha\in\F^k\mid \alpha_l=0\text{ for }l\in\cL_i\}.
\end{equation}
Now we are in a position to prove that KV-trellises are BCJR-trellises.

%%%%%%%%%%%%%%%%%%%%%%%%%%%%%%%%%%%%%
\begin{theo}\label{T-KVNS}
Let $(X,\cT)$ be a characteristic pair of~$\cC$ as in Definition~\ref{D-CharMat}.
Let $G\in\F^{k\times n}$ consist of~$k$ linearly independent rows of~$X$ and let $\cS$ be the list of
corresponding spans taken from~$\cT$.
Then the trellises $T_{G,\cS}$ and $T_{(G,H,\cS)}$ are isomorphic.
As a consequence, $\KV$-trellises are BCJR-trellises and thus non-mergeable.
\end{theo}
%%%%%%%%%%%%%%%%%%%%%%%%%%%%%%%%%%%%%

\begin{proof}
The last statement follows from Corollary~\ref{C-NSnonmerg}.
For the first part it suffices to prove that $\rk(M_i)=\rk(N_i)$ for all $i\in\cI$.
Indeed, then the result is immediate with Theorem~\ref{T-isomorphic}.
We will proceed in several steps.
\\
1) By Definition~\ref{D-CharMat}(iv) we may assume without loss of generality that the first~$k$ rows of~$X$ have linear span
(thus $a_l\leq b_l$ for $l=1,\ldots,k$) and the last $n-k$ rows have circular spans.
Then those first~$k$ rows $x_1,\ldots,x_k$ are linearly independent and form an MSGM of~$\cC$, see \cite[Def.~6.2, Thm.~6.11]{McE96}.
Define the matrix $Z\in\F^{(n-k)\times(n-k)}$ via
\begin{equation}\label{e-Z}
    Z=\begin{pmatrix} z_1\\ \vdots\\ z_{n-k}\end{pmatrix},\ \text{ where }
    z_l=\sum_{j=a_{k+l}}^{n-1}x_{k+l,j}H_j
\end{equation}
(where $x_{k+l}=(x_{k+l,0},\ldots,x_{k+l,n-1})$ is the $(k+l)$-th row of~$X$).
We claim that~$Z$ is nonsingular.
\\
In order to see this, assume $\beta Z=0$ for some $\beta=(\beta_1,\ldots,\beta_{n-k})\in\F^{n-k}$.
Hence
\begin{equation}\label{e-betaZ}
  \sum_{l=1}^{n-k}\beta_l\sum_{j=a_{k+l}}^{n-1}x_{k+l,j}H_j=0.
\end{equation}
For $l=1,\ldots,n-k$ define $\hat{x}_{k+l}:=(0,\ldots,0,x_{k+l,a_{k+l}},\ldots,x_{k+l,n-1})\in\F^n$, that is,
$\hat{x}_{k+l}$ is the vector with first $a_{k+l}$ entries equal to zero and the last ones identical to those of $x_{k+l}$.
Notice that by definition of the span of a vector $x_{k+l,a_{k+l}}\not=0$ and remember that $a_{k+1},\ldots,a_{n}$
are distinct by definition of a characteristic pair.
Therefore, $\hat{x}_{k+1},\ldots,\hat{x}_{n}$ are linearly independent.
Moreover,\eqnref{e-betaZ} implies
\[
  0=\sum_{l=1}^{n-k}\beta_l\sum_{j=0}^{n-1}\hat{x}_{k+l,j}H_j=\sum_{j=0}^{n-1}\big(\sum_{l=1}^{n-k}\beta_l\hat{x}_{k+l,j}\big)H_j.
\]
But this means that $\tilde{x}\in\ker H\T=\cC$, where $\tilde{x}:=\sum_{l=1}^{n-k}\beta_l\hat{x}_{k+l}$.
If we can show that $\tilde{x}=0$, then the linear independence of $\hat{x}_{k+1},\ldots,\hat{x}_{n}$
implies~$\beta=0$, which in turn proves that~$Z$ is nonsingular.
We proceed by contradiction and assume $\tilde{x}\not=0$.
Then $\tilde{x}$ has a unique linear span, say $(\tilde{a},\tilde{b}]$, where,
by definition of~$\tilde{x}$,
\begin{equation}\label{e-tildea1}
  \tilde{a}=\min\{a_{k+l}\mid \beta_l\not=0\}.
\end{equation}
On the other hand, the first~$k$ rows of~$X$ form an MSGM of ~$\cC$ and therefore the codeword~$\tilde{x}$ is of the form
$\tilde{x}=\sum_{l=1}^k\gamma_l x_l$
for some $\gamma_l\in\F$.
Now the predictable span property of an MSGM \cite[Lem.~6.6]{McE96} implies that
$\tilde{a}=\min\{a_l\mid \gamma_l\not=0\}$.
Along with\eqnref{e-tildea1} this contradicts the fact that $a_1,\ldots,a_{n}$ are distinct.
Thus $\tilde{x}=0$ as desired.
\\
2)
Consider now the matrix~$D$ as defined in\eqnref{e-DMat}.
The rows of~$D$ corresponding to linear spans are zero (see Proposition~\ref{P-NSspan}), while the
other rows, corresponding to circular spans, are certain rows of the matrix~$Z$ in\eqnref{e-Z}.
Hence the nonsingularity of~$Z$ implies that those rows are linearly independent.
All this shows that $\rk(D)=\rk(N_0)=\big|\{l \mid (a_l,b_l]\text{ circular}\}\big|=|\cL_0|=\rk(M_0)$, by virtue of
Proposition~\ref{P-SCPECP}.
\\
3)
From~2) we have $\rk(N_0)=\rk(M_0)$.
Now the identities $\rk(N_i)=\rk(M_i)$ for general~$i\in\cI$ follow from applying the shift~$\sigma$ and making use of
Remarks~\ref{R-shiftKV},~\ref{R-CharMat}(b), and~\ref{R-shiftNS}.
\end{proof}

Theorem~\ref{T-KVNS} shows that the set of $\KV$-trellises is a subset of the set of (one-to-one) BCJR-trellises.
The following example illustrates that this containment is proper, that is, the set of one-to-one BCJR-trellises is
strictly larger than the set of KV-trellises.
As a consequence, the class of non-mergeable one-to-one trellises is strictly larger than the class of KV-trellises.
At the end of this section we will sketch a proof showing that the class of non-mergeable (one-to-one) trellises coincides
with the class of (one-to-one) BCJR-trellises.

%%%%%%%%%%%%%%%%%%%%%%%%%%%%%%%%%
\begin{exa}\label{E-nonchiNS}
Consider
\[
    G=\begin{pmatrix}0&1&1&1&0\\1&0&0&1&0\\0&1&1&0&1\end{pmatrix}\in\F_2^{3\times 5}
    \text{ with span list }
    \cS=\big[(1,3],\, (3,0],\,(2,1]\big].
\]
Since the rows of~$G$ are linearly independent, the product trellis $T:=T_{G,\cS}$ is one-to-one.
Moreover, it has vertex spaces $V_i=\im M_i$, where
\[
  M_0=\!\begin{pmatrix}0\!&\!0\!&\!0\\0\!&\!1\!&\!0\\0\!&\!0\!&\!1\end{pmatrix}\!,\,
  M_1=\!\begin{pmatrix}0\!&\!0\!&\!0\\0\!&\!0\!&\!0\\0\!&\!0\!&\!1\end{pmatrix}\!,\,
  M_2=\!\begin{pmatrix}1\!&\!0\!&\!0\\0\!&\!0\!&\!0\\0\!&\!0\!&\!0\end{pmatrix}\!,\,
  M_3=\!\begin{pmatrix}1\!&\!0\!&\!0\\0\!&\!0\!&\!0\\0\!&\!0\!&\!1\end{pmatrix}\!,\,
  M_4=\!\begin{pmatrix}0\!&\!0\!&\!0\\0\!&\!1\!&\!0\\0\!&\!0\!&\!1\end{pmatrix}\!.
\]
Hence the SCP of~$T$ is given by $s=(2,1,1,2,2)$.
One can straightforwardly verify that a characteristic pair of the code $\cC=\im G\subseteq\F_2^5$ is given by
$(X,\cT)$ with span matrix~$S$, where
\[
  X=\begin{pmatrix}1&1&1&0&0\\0&1&1&1&0\\0&0&0&1&1\\1&0&0&0&1\\0&1&1&1&0\end{pmatrix},\;
  \cT=\big[(0,2],\,(1,3],\,(3,4],\,(4,0],\,(2,1]\big],\;
  S=\begin{pmatrix}0&1&1&0&0\\0&0&1&1&0\\0&0&0&0&1\\1&0&0&0&0\\1&1&0&1&1\end{pmatrix}.
\]
Using Remark~\ref{R-chi-trellis} one can see directly that no three rows of~$X$ and their according spans will lead to the
SCP above.
Hence~$T$ is not isomorphic to a $\KV$-trellis.
Let us use the parity check matrix
\[
   H=\begin{pmatrix}1&0&1&1&1\\0&1&1&0&0\end{pmatrix}
\]
in order to compute the BCJR-trellis $T_{(G,H,\cS)}$ of~$\cC$.
Applying\eqnref{e-DMat} and\eqnref{e-NMat} leads to
\[
  N_0=\begin{pmatrix}0&0\\1&0\\0&1\end{pmatrix},\;
  N_1=\begin{pmatrix}0&0\\0&0\\0&1\end{pmatrix},\;
  N_2=\begin{pmatrix}0&1\\0&0\\0&0\end{pmatrix},\;
  N_3=\begin{pmatrix}1&0\\0&0\\1&1\end{pmatrix},\;
  N_4=\begin{pmatrix}0&0\\1&0\\1&1\end{pmatrix}.
\]
Hence $\col(N_i)=\col(M_i)$ for all~$i\in\{0,\ldots,4\}$ and due to Theorem~\ref{T-isomorphic}  the trellis~$T$ is isomorphic to the BCJR-trellis
$T_{(G,H,\cS)}$ and, as a consequence,~$T$ is non-mergeable.
This can also be seen directly from the concrete trellis~$T_{(G,H,\cS)}$ in the figure below.
Let us also briefly verify the statements about the ECP as given in Proposition~\ref{P-SCPECP}.
The edge spaces $E_i$ are given by $E_i=\im(N_i,G_i,N_{i+1})$, where~$G_i$ is the $i$-th column of~$G$.
For sake of space we display the matrix $(N_0|G_0|N_1|\ldots|N_4|G_4|N_0)$.
It is given by
\[
   \left(\!\!\begin{array}{cc|c|cc|c|cc|c|cc|c|cc|c|cc}
   0&0&0&0&0&1&0&1&1&1&0&1&0&0&0&0&0\\
   1&0&1&0&0&0&0&0&0&0&0&1&1&0&0&1&0\\
   0&1&0&0&1&1&0&0&1&1&1&0&1&1&1&0&1\end{array}\!\!\right)
\]
and shows that the ECP is given by $e=(2,2,2,3,2)$.
Hence $e=s+(0,1,1,1,0)$ and the latter is indeed the indicator function of the set $\cA=\{1,2,3\}$ of
starting points of the spans.
Likewise, we have $e=(1,1,2,2,2)+(1,1,0,1,0)$ and the latter is the indicator function of $\cB=\{0,1,3\}$.
All this illustrates the formulas given in Proposition~\ref{P-SCPECP}.
\\
Finally, from the graph below it can easily be checked that the dual trellis $T_{(H,G,N_0^{\sf T})}$, representing $\cC^{\perp}=\im H$
(see Proposition~\ref{P-NSdual}), is mergeable at time $i=0$.
As a consequence, the dual of a BCJR-trellis need not be a BCJR-trellis.
Even more, the dual graph is disconnected and there does not exist a path from $(011)\in\im N_0^{\sf T}$ to $(000)\in\im N_0^{\sf T}$.
Hence the dual trellis does not satisfy Lemma~\ref{L-pathv0}.
\\
\mbox{}\hspace*{1cm} \includegraphics[height=3.7cm]{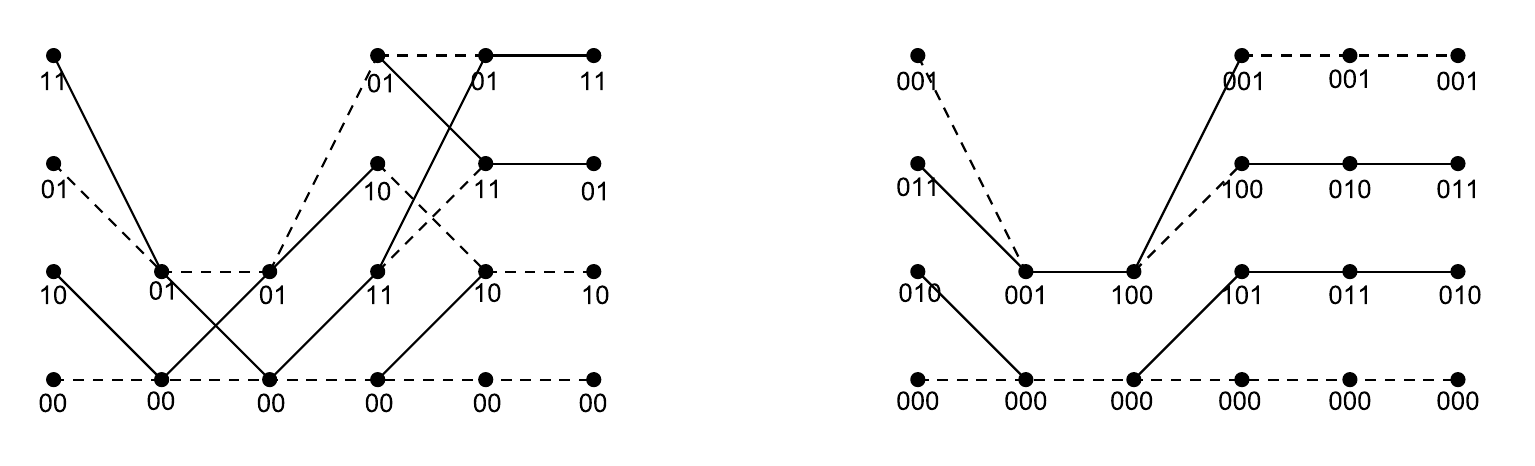}
     \\[-1ex]
\mbox{}\hspace*{2.4cm}{\tiny [Trellis~$T_{(G,H,\cS)}$]} \hspace*{5cm} {\tiny [Trellis~$T_{(H,G,N_0^{\sf T})}$]}
\end{exa}
%%%%%%%%%%%%%%%%%%%%%%%%%%%%%%%%

In \cite[Thm.~3.1]{NoSh06} it has been shown that a non-mergeable product trellis is isomorphic to
the corresponding BCJR-trellis.
The following theorem generalizes this result.
It shows that every (one-to-one) product trellis $T_{G,\cS}$ can be merged to the corresponding BCJR-trellis
and merging is accomplished by taking suitable quotients of the vertex spaces.
An instance of this process has been given in Example~\ref{E-BCJRnotone}.
It also reminds us to be aware of the fact that the resulting BCJR-trellis may not be one-to-one.

%%%%%%%%%%%%%%%%%%%%%%%%%%%%%
\begin{theo}\label{T-mergeProd}
Let~$\cS$ be as in\eqnref{e-S} and consider the trellises $T_{(G,H,\cS)}$ and $T_{G,\cS}$ along with the
matrices~$N_i$ and~$M_i$ as in\eqnref{e-DMat}/\!\!\eqnref{e-NMat} and\eqnref{e-MMat}, respectively.
Then
\begin{alphalist}
\item there exist subspaces~$W_i\subseteq\im M_i$ such that $\ker M_i\oplus W_i=\ker N_i$,
\item the quotient trellis~$T'$ of~$T_{G,\cS}$ defined as $T'=(V',E')$, where $V'=\cup_{i=0}^{n-1}V'_i$ and $E'=\cup_{i=0}^{n-1}E'_i$ with
      \[
          V'_i= \im M_i/W_i,\quad E'_i=\{(\alpha M_i+W_i,\,\alpha G_i,\alpha M_{i+1}+W_{i+1})\mid \alpha\in\F^k\},
      \]
      is isomorphic to the BCJR-trellis $T_{(G,H,\cS)}$.
      In particular, $\cC(T')=\cC$.
\end{alphalist}
\end{theo}
%%%%%%%%%%%%%%%%%%%%%%%%%%%%%%

\begin{proof}
(a) First notice that by the very definition of~$M_i$ in\eqnref{e-MMat} along with Proposition~\ref{P-NSspan}(b) we have $M_iN_i=N_i$ and thus
$\ker M_i\subseteq\ker N_i$.
Hence there exists a subspace $W_i'\subseteq\F^k$ such that $\ker M_i\oplus W'_i=\ker N_i$.
Define the subspace $W_i:=W'_iM_i:=\{\alpha M_i\mid \alpha\in W'_i\}\subseteq\im M_i$.
Notice that $(M_i)^2=M_i$, thus $M_i$ defines a projection and hence $\ker M_i\oplus \im M_i=\F^k$.
This implies $\ker M_i\cap W_i=\{0\}$.
The inclusion $W_i\subseteq\ker N_i$ follows immediately from the identity $M_iN_i=N_i$ along with the fact that $W'_i\subseteq\ker N_i$.
Now let $\beta\in\ker N_i$. Then there exist $\alpha\in\ker M_i$ and $\gamma\in W'_i$ such that
$\beta=\alpha+\gamma$.
But then $\beta=\alpha+\gamma(I-M_i)+\gamma M_i$ and $\gamma M_i\in W_i$ while $\alpha+\gamma(I-M_i)\in\ker M_i$
due to $(I-M_i)M_i=0$.
All this proves~(a).
\\
(b) First of all, by virtue of~(a) we have
$\dim V'_i=\dim(\im M_i)-\dim W_i=k-(\dim\ker M_i+\dim W_i)=k-\dim\ker N_i=\rk N_i$.
Moreover, it is easy to see that
\[
  \psi_i:\; V'_i\longrightarrow \im N_i,\quad \alpha M_i+W_i\longmapsto\alpha N_i
\]
is a well-defined, surjective homomorphism and thus an isomorphism between the vertex spaces $V'_i$ of~$T'$ and
$\im N_i$ of $T_{(G,H,\cS)}$.
But then it is clear from the very definition of the edge spaces $E'_i$ above and $E_i$ as in Theorem~\ref{T-NS}
that $T'$ and $T_{(G,H,\cS)}$ are isomorphic as defined in~(e) in the introduction.
\end{proof}

We close this section with the following observation.

%%%%%%%%%%%%%%%%%%%%%%%%%%%%%%%
\begin{rem}\label{R-nonmerg}
In this remark we sketch a proof showing that the class of linear\footnote{Recall from Definition~\ref{D-TBTbasics}(a)
that linear trellises are reduced.}, non-mergeable (not necessarily one-to-one)
trellises coincides with the class of BCJR-trellises as in Definition~\ref{D-NSspan}, but where the rows of~$G$ maybe linearly dependent.
\\
First of all, let $G\in\F^{k\times n}$ where $\rk G<k$.
It not hard to show that Theorem~\ref{T-KVtrellis} remains valid (with the exception that~$T_{G,\cS}$ is not one-to-one anymore)
and that BCJR-trellises as in Theorem~\ref{T-NS} and~\ref{D-NSspan} can be defined in the same way.
The resulting trellis $T_{(G,H,\cS)}$ still represents~$\cC=\im G=\ker H\T$.
Moreover, one can easily verify that Propositions~\ref{P-edgepath},~\ref{P-NSspan}, Corollary~\ref{C-NSnonmerg},
Theorem~\ref{T-isomorphic}, and Theorem~\ref{T-mergeProd} remain valid with the same proofs (with the exception
of part~\ref{T-isomorphic}(b); see also the footnote on Page~\pageref{PageFootnote}).
For sake of clarity let us call the resulting BCJR-trellises genBCJR-trellises.
\\
Now the above implies in particular that every genBCJR-trellis is non-mergeable.
For the converse, let~$T$ be a linear and non-mergeable trellis of~$\cC$.
By \cite[Thm.~4.2]{KoVa03}~$T$ is the product of elementary trellises $T_{x_l,(a_l,b_l]}, l=1,\ldots,m$, where
$\{x_1,\ldots, x_m\}$ is a (not necessarily linearly independent) generating set of~$\cC$ and where
$(a_l,b_l]$ is a generalized span of~$x_l$, that is, $x_{l,j}=0$ for $j\not\in[a_l,b_l]$ (notice the difference to our
definition of a span in Def.~\ref{D-vectorspan}).
By \cite[Lem.~4.3]{KoVa03} each factor $T_{x_l,(a_l,b_l]}$ is non-mergeable.
But then it is easy to see that $(a_l,b_l]$ is a span of~$x_l$ in the sense of Def.~\ref{D-vectorspan} (that is, $x_l$ is nonzero
at the endpoints~$a_l$ and~$b_l$) for otherwise the elementary trellis could be merged at time~$a_l$ or~$b_l$.
Hence $T=T_{G,\,\cS}$ in the sense of Def.~\ref{D-KVtrellis} but where the rows of~$G$ are possibly linearly dependent.
By Theorem~\ref{T-mergeProd}~$T$ can be merged to the corresponding genBCJR-trellis, hence, by non-mergeability,~$T$
is a genBCJR-trellis.
As a consequence, if~$T$ is one-to-one, then~$T$ is a BCJR-trellis.
\end{rem}
%%%%%%%%%%%%%%%%%%%%%%%%%%%%%%%%

%%%%%%%%%%%%%%%%%%%%%%%%%%%%%%%
\Section{On a Duality Conjecture by Koetter/Vardy}\label{S-DualSelect}
%%%%%%%%%%%%%%%%%%%%%%%%%%%%%%%

Koetter/Vardy have shown the following.

%%%%%%%%%%%%%%%%%%%%%%%%%%
\begin{lemma}[\mbox{\cite[Lem.~5.11]{KoVa03}}]\label{L-DualCharMat}
Let $\cC\subseteq\F^n$ and $\cC^{\perp}$ both be codes with support~$\cI$ and
let the characteristic span list of $\cC$ be given by $\cT=\big[(a_l,b_l],\,l=1,\ldots,n]$.
Then the characteristic span list of $\cC^{\perp}$ is given by $\big[(b_l,a_l],\,l=1,\ldots,n]$.
\end{lemma}
%%%%%%%%%%%%%%%%%%%%%%%%%%%%%%
Moreover, in \cite[Prop.~5.13]{KoVa03} they proved the following:
select~$k$ linearly independent rows of a characteristic matrix of~$\cC$, with spans, say,
$(a_{l_i},b_{l_i}],\,i=1,\ldots,k$, and construct the resulting KV-trellis~$T$.
Pick the $n-k$ rows of the characteristic matrix of~$\cC^{\perp}$ that do \underline{not} have spans $(b_{l_i},a_{l_i}],\,i=1,\ldots,k$
and construct the resulting product trellis~$\hat{T}$.
Then the trellises~$T$ and~$\hat{T}$ share the same SCP.
However, while the trellis~$T$ for~$\cC$ is, by construction, a KV-trellis in the sense of Definition~\ref{D-KVtrellis},
the trellis~$\hat{T}$ might not even represent~$\cC^{\perp}$ because it is not a priori clear whether
the underlying~$n-k$ rows are linearly independent.
Koetter/Vardy conjecture that those rows are indeed always linearly independent and thus the trellis~$\hat{T}$ is a KV-trellis
for~$\cC^{\perp}$.
For all this one has to have in mind that Koetter/Vardy define the characteristic matrix in a way which makes it
unique for a given code and their conjecture makes sense as phrased.
However, in our definition a characteristic matrix is not unique and therefore the linear independence of the selected rows
might even depend on the choice of that matrix.
This is indeed the case as the following example shows.
%%%%%%%%%%%%%%%%%%%%%%%%%%%%%%%
\begin{exa}\label{E-dualchoice}
Consider again Example~\ref{E-chitrellisnonunique}.
The underlying binary code and its dual are
$\cC=\{0000,\,1100,\,0111,\,1011\}$ and $\cC^{\perp}=\{0000,\,1101,\,0011,\,1110\}$.
A characteristic pair for the dual code~$\cC^{\perp}$ is given by
\[
   (\hat{X},\,\hat{\cT})=\Big(\begin{pmatrix}1&1&0&1\\1&1&0&1\\1&1&1&0\\0&0&1&1\end{pmatrix},\,
                            \big[(1,0],\,(3,1],\,(0,2],\,(2,3]\big]\Big).
\]
In Example~\ref{E-chitrellisnonunique} we observed that the last two rows of~$X$ are linearly independent and thus
give rise to a KV-trellis of~$\cC$.
The according spans are $(2,0]$ and $(3,2]$.
Hence the complementary rows for~$\cC^{\perp}$ are the first two rows of~$\hat{X}$, but they are linearly dependent.
Replacing the first row of~$\hat{X}$ with the codeword $(1,1,1,0)\in\cC^{\perp}$ and keeping the span $(1,0]$, one obtains another characteristic
matrix for~$\cC^{\perp}$ for which that row selection is indeed linearly independent.
All this shows that the conjecture of Koetter/Vardy does depend on the choice of the characteristic matrix.
One might also want to observe that both~$X$ and~$\hat{X}$ are the lexicographically first characteristic matrix of their respective code and
therefore the ones Koetter/Vardy singled out in their approach.
Thus the example also shows that their conjecture does not hold for this particular choice.
\end{exa}
%%%%%%%%%%%%%%%%%%%%%%%%%%%%%%

For the specific classes of self-dual codes and cyclic codes the Koetter/Vardy conjecture has been proven in~\cite{KaSh05}.

In the sequel we will show that for row selections leading to minimal
KV-trellises (in the sense of Definition~\ref{D-TBTmin}(a)) there always exists a characteristic matrix of~$\cC^{\perp}$ such that
Koetter/Vardy's conjecture holds true.

%%%%%%%%%%%%%%%%%%%%%%%%%%%%%
\begin{theo}\label{T-dualselect}
Let $\cC,\,\cC^{\perp}\subseteq\F^n$ both have support~$\cI$.
Furthermore, let~$(X,\cT)$ be a characteristic pair of~$\cC$, where $\cT=[(a_l,b_l],l=1,\ldots,n]$.
Suppose the rows of~$X$ with indices $l_1,\ldots,l_k$ are linearly independent and give rise to a minimal
KV-trellis.
Then there exists a characteristic matrix~$\hat{X}$ of~$\cC^{\perp}$ such that
the $n-k$ rows of~$\hat{X}$ that do not have spans $(b_{l_i},a_{l_i}],i=1\ldots,k$,
are linearly independent and give rise to a minimal KV-trellis of~$\cC^{\perp}$.
\end{theo}
%%%%%%%%%%%%%%%%%%%%%%%%%%%%%%
In the proof we will use a dualization technique for trellises going back to Mittelholzer~\cite{Mi95} and
Forney~\cite[Sec.~VII.D and Thm.~8.4]{Fo01}.
Their construction of a dual trellis amounts to dualizing the edge spaces (called local constraints in~\cite{Fo01})
with respect to a particular bilinear form.
In our case this bilinear form reads as in\eqnref{e-bilinearform} below.
As the proof will show for minimal BCJR-trellises the resulting dual trellis coincides with the BCJR-dual defined in
Proposition~\ref{P-NSdual}.

\begin{proof}
Without loss of generality let us assume $(l_1,\ldots,l_k)=(1,\ldots,k)$ and let $G\in\F^{k\times n}$
be the matrix consisting of the first~$k$ rows of~$X$.
Then~$G$ is a generator matrix of~$\cC$.
As before, let the parity check matrix~$H$ be as in\eqnref{e-Hdata}.
By Theorem~\ref{T-KVNS} we know that the corresponding KV-trellis $T_{G,\cS}$ is isomorphic to the BCJR-trellis $T:=T_{(G,H,\cS)}$,
where $\cS=\big[(a_l,b_l],l=1,\ldots,k\big]$.
Denote the sets of starting and end points by $\cA=\{a_1,\ldots,a_k\}$ and $\cB=\{b_1,\ldots,b_k\}$.
The vertex and edge spaces of~$T$ are given by $\im N_i$ and $E_i:=\im(N_i,G_i,N_{i+1}),\,i\in\cI$,
where~$N_i$ is as in\eqnref{e-NMat},\eqnref{e-DMat}.
Let $s_i:=\rk(N_i)$, thus $(s_0,\ldots,s_{n-1})$ is the SCP of~$T$.

By Proposition~\ref{P-NSdual} the BCJR-trellis $\hat{T}:=T_{(H,G,N_0^{\sf T})}$ represents~$\cC^{\perp}$ and has the same SCP.
Its vertex and edge spaces are given by
$\im\hat{N}_i$ and $\hat{E}_i:=\im(\hat{N}_i,H_i\T,\hat{N}_{i+1})$, where $\hat{N}_i=N_i\T$.
Since~$T$ is a minimal trellis the same is true for~$\hat{T}$; see also \cite[Thm.~4.1]{NoSh06}.
Therefore, by virtue of Theorem~\ref{T-min-trellis}, the trellis~$\hat{T}$ must be a KV-trellis for~$\cC^{\perp}$.
This means, there exists a characteristic pair~$(\hat{X},\,\hat{\cT})$ of~$\cC^{\perp}$ such that
$\hat{T}$ is a $\KV_{(\hat{X},\,\hat{\cT})}$-trellis, which
in turn implies that there exists a selection of~$n-k$ rows of~$\hat{X}$ giving rise to the trellis~$\hat{T}$.
Let~$\hat{\cS}$ be the list of associated characteristic spans and let~$\hat{\cA}$ be the set of
starting points of~$\hat{\cS}$.
Thus, $|\hat{\cA}|=n-k$.

By virtue of Lemma~\ref{L-DualCharMat} we have to show that $\hat{\cA}=\{b_{k+1},\ldots,b_n\}$.
From Theorem~\ref{T-isomorphic} and Proposition~\ref{P-SCPECP} we know
\begin{equation}\label{e-ei}
   \hat{e}_i:=\dim\hat{E}_i=s_i+I^{\hat{\cA}}_i\ \text{ and }e_i:=\dim E_i=s_{i+1}+I^{\cB}_i,\,i\in\cI.
\end{equation}
Using that $\hat{N}_i=N_i\T$ one obtains the following well-defined non-degenerate bilinear form
\begin{equation}\label{e-bilinearform}
\begin{split}
  \Pi:\;(\im N_i\times\F\times\im N_{i+1})\times(\im\hat{N}_i\times\F\times\im\hat{N}_{i+1})&\longrightarrow\qquad \F\\
       \big((\alpha N_i,a,\tilde{\alpha}N_{i+1}),\,(\beta\hat{N}_i,b,\tilde{\beta}\hat{N}_{i+1})\big)\qquad&\longmapsto
          \alpha N_i\beta\T+ab-\tilde{\alpha}N_{i+1}\tilde{\beta}\T.
\end{split}
\end{equation}
Then we have for all $\alpha\in\F^k$ and $\beta\in\F^{n-k}$ and $i\in\cI$
\begin{align*}
  \Pi\big((\alpha N_i,\alpha G_i,\alpha N_{i+1}),\,(\beta\hat{N}_i,\beta H_i\T,\beta\hat{N}_{i+1})\big)
  &=\alpha N_i\beta\T+\alpha G_i H_i\beta\T-\alpha N_{i+1}\beta\T\\
  &=\alpha(N_i+G_i H_i-N_{i+1})\beta\T=0,
\end{align*}
where the last identity follows from the very definition of $N_{i+1}$ in\eqnref{e-NMat}.
Hence $\hat{E_i}\subseteq (E_i)^\circ$, where $(E_i)^{\circ}$ denotes the dual space of~$E_i$ in $\im\hat{N}_i\times\F\times\im\hat{N}_{i+1}$
with respect to the bilinear form~$\Pi$.
Using that $\dim(\im\hat{N}_i\times\F\times\im\hat{N}_{i+1})=\dim(\im N_i\times\F\times\im N_{i+1})=s_i+s_{i+1}+1$ and
applying the second identity in\eqnref{e-ei}, we conclude
\begin{equation}\label{e-ehate}
  \hat{e}_i=\dim\hat{E}_i\leq\dim(E_i)^{\circ}= s_i+s_{i+1}+1-e_i=s_i+1-I^{\cB}_i.
\end{equation}
Along with the first identity in\eqnref{e-ei} this yields
\begin{equation}\label{e-indfunc}
   I^{\hat{\cA}}_i+I^{\cB}_i\leq 1 \text{ for all }i\in\cI
\end{equation}
for the indicator functions on~$\hat{\cA}$ and~$\cB$.
Hence the sets $\hat{\cA}$ and~$\cB$ are disjoint and since the two sets are subsets of~$\cI$ of cardinality~$n-k$ and~$k$,
respectively, we may conclude $\hat{\cA}=\cI\,\backslash\,\cB=\{b_{k+1},\ldots,b_n\}$ as desired.
This completes the proof.
\\
One should notice that we now also have equality in\eqnref{e-indfunc} and\eqnref{e-ehate} and thus derived the identity $\hat{E_i}= (E_i)^\circ$,
which is an instance of \cite[Cor.~8.2]{Fo01}.
\end{proof}

It is worth noting that the proof does not explicitly require the minimality of the trellis, but only that the dual
trellis $\hat{T}=T_{(H,G,N_0^{\sf T})}$ be a KV-trellis.
We strongly believe that Theorem~\ref{T-dualselect} is true for all KV-trellises (and not just minimal ones).
Moreover, we believe that the BCJR-dual of a KV-trellis is isomorphic to the trellis obtained from
the dual row selection of a suitably chosen dual characteristic matrix.

%%%%%%%%%%%%%%%%%%%%%%%%
\begin{rem}\label{R-localdual}
The proof above shows that for KV-trellises whose BCJR-dual is a KV-trellis again, this BCJR-dual coincides
with the dual obtained by dualizing the edge spaces.
In other words, $\hat{E}_i=(E_i)^{\circ}$ for $i\in\cI$ in the notation of the proof above.
We wish to point out that this is not always the case for BCJR-trellises, not even if they are one-to-one and
isomorphic to the corresponding product trellis (recall that BCJR-trellises are always mergeable).
An example can be drawn from Example~\ref{E-nonchiNS}.
The trellis $T_{G,\cS}\cong T_{(G,H,\cS)}$ has SCP $(2,1,1,2,2)$ and ECP $(2,2,2,3,2)$.
The dual trellis $\hat{T}=T_{(H,G,N_0^{\sf T})}$ has, of course, the same SCP and its ECP is $(2,1,2,2,2)$.
According to\eqnref{e-ehate} the trellis~$T^{\circ}$ with vertex spaces $\im\hat{N}_i$
and edge spaces $(E_i)^{\circ}$ has ECP $(2,1,2,2,3)$, hence $\hat{E}_4\subsetneq(E_4)^{\circ}$
and $\hat{E}_i=(E_i)^{\circ}$ for $i=0,\ldots,3$.
One computes
\[
   (E_4)^{\circ}=\im\left(\!\!\begin{array}{ccc|c|ccc}
           0&1&1&1&0&1&0\\0&0&1&0&0&0&1\\0&0&0&1&0&0&1\end{array}\!\!\right)\in\im\hat{N}_4\times\F\times\im\hat{N}_0,
\]
which results in the following trellis
\begin{center}
    \includegraphics[height=3.7cm]{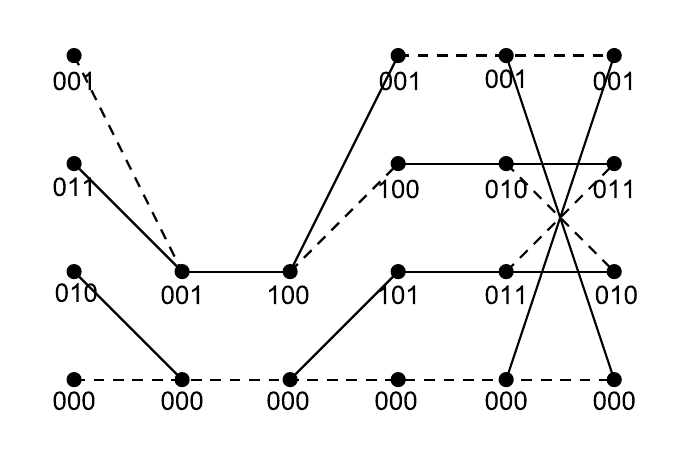}
    \\[-1ex]
    {\tiny [Trellis~$T^{\circ}$]}
\end{center}
As one can easily see, this trellis does indeed represent~$\cC^{\perp}$, is one-to-one, but it is not reduced!
As a consequence, it is not considered a linear trellis due to Definition~\ref{D-TBTbasics}(a).
Removing the~$4$ cross-diagonal edges in~$(E_4)^{\circ}$ results in the BCJR-dual $T_{(H,G,N_0^{\sf T})}$
displayed in Example~\ref{E-nonchiNS}.
\end{rem}
%%%%%%%%%%%%%%%%%%%%%%%

\section*{Some Open Problems}
\begin{liste}
\item In Theorem~\ref{T-min-trellis} we have seen that every minimal trellis is a $\KV_{(X,\cT)}$-trellis for a particular choice
      of the characteristic pair $(X,\cT)$.
      Example~\ref{E-CMmintrellis} showed that indeed a suitable choice of the characteristic matrix is needed for this result to be true.
      In that particular example a different choice of the characteristic matrix leads to a trellis that is (only) structurally isomorphic to the given trellis.
      All this gives rise to the question whether every minimal trellis arises, up to structural isomorphism, from a fixed characteristic matrix?
      It is easy to see that this amounts to investigating whether if a certain row selection of one characteristic matrix is linearly independent and
      gives rise to a minimal trellis of~$\cC$, then the same row selection is linearly independent for each characteristic matrix of the code.
      If this turns out to be true then this implies that in Theorem~\ref{T-dualselect} the statement would be true for each characteristic
      matrix of the dual code.
\item Does Theorem~\ref{T-dualselect} hold true for all KV-trellises? This would be answered in the affirmative if one can show that the BCJR-dual
      of a KV-trellis is a KV-trellis again.
\item From Examples~\ref{E-chitrellisnonunique}/\ref{E-dualchoice} it is clear that the duality result of
      Theorem~\ref{T-dualselect} does indeed depend on the choice of the characteristic matrix of the dual code
      (at least, if the trellises are not minimal).
      However, several examples suggest that one may ask whether there exist a dual pairing of characteristic matrices~$X$ and~$\hat{X}$
      for~$\cC$ and~$\cC^{\perp}$, respectively, such that Koetter/Vardy's duality conjecture holds true for these pairs,
      that is, such that the Koetter/Vardy-dual of a $\KV_{(X,{\mathcal T})}$-trellis is a $\KV_{(\hat{X},\hat{\mathcal T})}$-trellis.
      For instance, for Examples~\ref{E-chitrellisnonunique}/\ref{E-dualchoice} this is true for the pairing $(X',\,\hat{X})$ and  when
      pairing~$X$ with the (unique) remaining characteristic matrix of~$\cC^{\perp}$.
\item Compare the BCJR-dualization with the dualization of trellises via dualizing the edge spaces as in~\cite{Mi95,Fo01}.
      More specifically, under which conditions do these dual trellises coincide when using
      the bilinear form\eqnref{e-bilinearform}?
\end{liste}

%%%%%%%%%%%%%%%%%%%%%%%%%%%%%%%%%%%%%%%%%%%%%%
\bibliographystyle{abbrv}

\begin{thebibliography}{10}

\bibitem{BCJR74}
L.~R. Bahl, J.~Cocke, F.~Jelinek, and J.~Raviv.
\newblock Optimal decoding of linear codes for minimizing symbol error rate.
\newblock {\em IEEE Trans. Inform. Theory}, IT-20:284--287, 1974.

\bibitem{CFV99}
A.~R. Calderbank, {G.~D. Forney, Jr.}, and A.~Vardy.
\newblock Minimal tail-biting trellises: {T}he {G}olay code and more.
\newblock {\em IEEE Trans. Inform. Theory}, IT-45:1435--1455, 1999.

\bibitem{Fo09}
{G.~D. Forney, Jr.}
\newblock Minimal realizations of linear systems: {T}he ``shortest {B}asis''
  approach.
\newblock Preprint 2009. ArXiv: 0910.4336.

\bibitem{Fo88}
{G.~D. Forney, Jr.}
\newblock Coset codes~{II}: {B}inary lattices and related codes.
\newblock {\em IEEE Trans. Inform. Theory}, IT-34:1152--1187, 1988.

\bibitem{Fo01}
{G.~D. Forney, Jr.}
\newblock Codes on graphs: {N}ormal realizations.
\newblock {\em IEEE Trans. Inform. Theory}, IT-47:520--548, 2001.

\bibitem{FT93}
{G.~D. Forney, Jr.} and M.~D. Trott.
\newblock The dynamics of group codes: {S}tate spaces, trellis diagrams, and
  canonical encoders.
\newblock {\em IEEE Trans. Inform. Theory}, IT-39:1491--1513, 1993.

\bibitem{KaSh05}
H.~Kan and H.~Shen.
\newblock A relation between the characteristic generators of a linear code and
  its dual.
\newblock {\em IEEE Trans. Inform. Theory}, IT-51:1199--1202, 2005.

\bibitem{KoVa02}
R.~Koetter and A.~Vardy.
\newblock On the theory of linear trellises.
\newblock In M.~Blaum, P.~G. Farrel, and H.~C.~A. van Tilborg, editors, {\em
  Information, Coding and Mathematics}, pages 323--354. Kluwer, Boston, 2002.

\bibitem{KoVa03}
R.~Koetter and A.~Vardy.
\newblock The structure of tail-biting trellises: {M}inimality and basic
  principles.
\newblock {\em IEEE Trans. Inform. Theory}, IT-49:2081--2105, 2003.

\bibitem{KschSo95}
F.~R. Kschischang and V.~Sorokine.
\newblock On the trellis structure of block codes.
\newblock {\em IEEE Trans. Inform. Theory}, IT-41:1924--1937, 1995.

\bibitem{LiSh00}
S.~Lin and R.~Y. Shao.
\newblock General structure and construction of tail-biting trellises for
  linear block codes.
\newblock In {\em Proceedings of the 2000 IEEE International Symposium on
  Information Theory}, page 117, 2000.

\bibitem{McE96}
R.~J. Mc{E}liece.
\newblock On the {BCJR} {T}rellis for linear block codes.
\newblock {\em IEEE Trans. Inform. Theory}, IT-42:1072--1092, 1996.

\bibitem{Mi95}
T.~Mittelholzer.
\newblock Convolutional codes over groups: {A} pragmatic approach.
\newblock In {\em Proc. of the 33rd Allerton Conference on Communications,
  Control, and Computing}, pages 380--381, 1995.

\bibitem{Mu88}
D.~J. Muder.
\newblock Minimal trellises for block codes.
\newblock {\em IEEE Trans. Inform. Theory}, IT-34:1049--1053, 1988.

\bibitem{NoSh06}
A.~V. Nori and P.~Shankar.
\newblock Unifying views of tail-biting trellis constructions for linear block
  codes.
\newblock {\em IEEE Trans. Inform. Theory}, IT-52:4431--4443, 2006.

\bibitem{PoWi98}
J.~W. Polderman and J.~C. Willems.
\newblock {\em Introduction to Mathematical Systems Theory; {A} behavioral
  approach}.
\newblock Springer, New York, 1998.

\bibitem{ShBe00}
Y.~Shany and Y.~Be'ery.
\newblock Linear tail-biting trellises, the square root bound, and applications
  for {R}eed-{M}uller codes.
\newblock {\em IEEE Trans. Inform. Theory}, IT-46:1514--1523, 2000.

\bibitem{Va98}
A.~Vardy.
\newblock Trellis structure of codes.
\newblock In V.~S. Pless and W.~C. Huffman, editors, {\em Handbook of Coding
  Theory, Vol.~2}, pages 1989--2117. Elsevier, Amsterdam, 1998.

\bibitem{Wi89}
J.~C. Willems.
\newblock Models for dynamics.
\newblock {\em Dynamics Reported}, 2:171--269, 1989.

\bibitem{Wo78}
J.~K. Wolf.
\newblock Efficient maximum likelihood decoding of linear block codes.
\newblock {\em IEEE Trans. Inform. Theory}, IT-24:76--80, 1978.

\end{thebibliography}

%%%%%%%%%%%%%%%%%%%%%%%%%%%%%%%%%%%%%%%%%%%%%%
\end{document}